\documentclass[letterpaper,11pt]{article}
\usepackage[margin=1in]{geometry}
\usepackage{amsfonts,amsmath,amssymb,amsthm}
\usepackage{thmtools} \usepackage{thm-restate}
\usepackage[noadjust]{cite}
\usepackage{comment,enumerate,hyperref}
\usepackage[colorinlistoftodos,textsize=small,color=red!25!white,obeyFinal]{todonotes}
\newtheorem{theorem}{Theorem}[section]

\newtheorem{lemma}[theorem]{Lemma}

\newtheorem{conjecture}[theorem]{Conjecture}

\usepackage{soul}
\theoremstyle{definition}
\newtheorem{definition}{Definition}

\usepackage{xcolor,xspace}
\usepackage{booktabs}

\usepackage[boxruled, noend]{algorithm2e}
\usepackage{placeins}

\usepackage{csquotes}

\newcommand{\poly}{\text{poly}}
\newcommand{\dist}{\text{dist}}
\newcommand{\Oish}{\widetilde{O}}
\newcommand{\Omegaish}{\widetilde{\Omega}}

\newcommand{\ffs}{\texttt{find-fault-set}}
\newcommand{\eps}{\varepsilon}

\DeclareMathOperator*{\E}{\mathbb{E}}

\title{Vertex Fault-Tolerant Emulators}

 \author{Greg Bodwin\\ University of Michigan\\ bodwin@umich.edu \and
Michael Dinitz\thanks{Supported in part by NSF award CCF-1909111.} \\ Johns Hopkins University\\ mdinitz@cs.jhu.edu\and
Yasamin Nazari \thanks{Supported in part by NSF award CCF-1909111 and by Austrian Science Fund (FWF) grant P 32863-N.}\\ University of Salzburg\\ ynazari@cs.sbg.ac.at}

\date{}

\begin{document}
\begin{titlepage}
\maketitle

\begin{abstract} 
    A \emph{$k$-spanner} of a graph $G$ is a sparse subgraph that preserves its shortest path distances up to a multiplicative stretch factor of $k$, and a \emph{$k$-emulator} is similar but not required to be a subgraph of $G$.
    A classic theorem by Alth{\" o}fer et al.~[Disc.\ Comp.\ Geom.\ '93] and Thorup and Zwick [JACM '05] shows that, despite the extra flexibility available to emulators, the size/stretch tradeoffs for spanners and emulators are equivalent.
    Our main result is that this equivalence in tradeoffs no longer holds in the commonly-studied setting of graphs with vertex failures.
    That is: we introduce a natural definition of vertex fault-tolerant emulators, and then we show a three-way tradeoff between size, stretch, and fault-tolerance for these emulators that polynomially surpasses the tradeoff known to be optimal for spanners.
    
    We complement our emulator upper bound with a lower bound construction that is essentially tight (within $\log n$ factors of the upper bound) when the stretch is $2k-1$ and $k$ is either a fixed odd integer or $2$.
    We also show constructions of fault-tolerant emulators with additive error, demonstrating that these also enjoy significantly improved tradeoffs over those available for fault-tolerant additive spanners.
\end{abstract}

\thispagestyle{empty}

\end{titlepage}

\renewcommand{\arraystretch}{1.5} 

\section{Introduction}
Two well-studied objects in graph sparsification are \emph{spanners} and \emph{emulators}.  Given a weighted input graph $G = (V,E, w)$, a $t$-spanner of $G$ is a subgraph $H$ of $G$ in which 
\begin{equation} \label{eq:stretch}
\dist_G(u,v) \leq \dist_H(u,v) \leq t \cdot \dist_H(u,v)
\end{equation}
for all $u,v \in V$.  Note that the first inequality, that $\dist_G(u,v) \leq \dist_H(u,v)$, is implied automatically by the fact that $H$ is a subgraph of $G$.  The value $t$ is called the the \emph{stretch} of the spanner.

A $t$-emulator \cite{dor2000all} is defined in the same way, except that $H$ is not required to be a subgraph of $G$.
For emulators, the first inequality is not automatic, and it implies that any edge $(u,v)$ in the emulator $H$ but not in the input graph $G$ must have weight at least $\dist_G(u,v)$.
In fact, it is easy to see that without loss of generality that we may assign it weight exactly $\dist_G(u,v)$.

\subsection{Equivalence of Spanner/Emulator Tradeoffs}

Both spanners and emulators have been studied extensively, and we have long had a complete understanding of the tradeoffs between spanner/emulator size (number of edges) and stretch.
Specifically:
\begin{itemize}
    \item Alth\"{o}fer et al.~\cite{AlthoferDDJS:93} proved that for every positive integer $k$, every weighted graph $G = (V, E, w)$ has a $(2k-1)$-spanner with at most $O(n^{1+1/k})$ edges.

    \item On the lower bounds side, one can quickly verify that any unweighted input graph $G$ of girth $>2k$ has no $(2k-1)$ spanner, except for $G$ itself.
    Under the Erd\H{o}s girth conjecture~\cite{erdHos1964extremal}, there are graphs of girth $>2k$ with $\Omega(n^{1+1/k})$ edges.
    Thus, the upper bound of Alth\"{o}fer et al.\ cannot be improved at all on these graphs.

    \item Alth{\" o}fer et al.~\cite{AlthoferDDJS:93} and Thorup and Zwick \cite{thorup2005approximate} observed that essentially the same lower bound applies for emulators.
    For any two subgraphs $H_1, H_2$ of a graph of girth $>2k$, they disagree on some pairwise distance by more than $\cdot (2k-1)$.
    This implies that $H_1, H_2$ need different representations as $(2k-1)$-emulators.
    There are $2^{\Omega(n^{1+1/k})}$ subgraphs of a girth conjecture graph, and so by a pigeonhole argument, one of these subgraphs requires an emulator on $\Omegaish(n^{1+1/k})$ edges.
    (In fact, the same method gives a lower bound on the size of any \emph{data structure} that approximately encodes graph distances, and hence this is often called an \emph{incompressibility argument}.)
\end{itemize}

Thus, even though emulators are substantially more general objects than spanners, they do not enjoy a meaningfully better tradeoff between size and stretch.

\subsection{Fault-Tolerant Spanners}

Spanners are commonly applied as a primitive in distributed computing, in which network nodes or edges are prone to sporadic failures.
This has motivated significant interest in \emph{fault-tolerant} spanners.
Intuitively, a vertex fault-tolerant spanner is a subgraph that remains a spanner even after a small set of nodes fails in both the spanner and the original graph.
More formally, the following definition was given by Chechik, Langberg, Peleg, and Roditty~\cite{ChechikLPR:10}.

\begin{definition}[VFT Spanner] \label{def:ftspanner}
Let $G = (V, E, w)$ be a weighted graph.  A subgraph $H$ of $G$ is an $f$-vertex fault-tolerant ($f$-VFT) $t$-spanner of $G$ if, for all $F \subseteq V$ with $|F| \le f$, $H \setminus F$ is a $t$-spanner of $G \setminus F$.
\end{definition}

After significant work following~\cite{ChechikLPR:10}, we now completely understand the achievable bounds on fault-tolerant spanners: Bodwin and Patel~\cite{BP19} proved that every graph has an $f$-VFT $(2k-1)$ spanner with at most $O\left(f^{1-1/k} n^{1+1/k}\right)$ edges (and the same bounds were shown to be achievable in polynomial time by~\cite{DR20,BDR21}), and Bodwin, Dinitz, Parter, and Williams~\cite{BDPW18} gave examples (under the girth conjecture) of graphs on which this bound cannot be improved in any range of parameters.\footnote{We note that there is a related definition of \emph{edge} fault-tolerant spanners which has also been studied extensively, and which admits different bounds~\cite{BDR22}.  But in this paper we focus on the vertex case and do not study edge fault-tolerant emulators, leaving that as an interesting direction for future research.}

\subsection{Fault-Tolerant Emulators}

In this paper we ask a natural question: what if we add a fault-tolerance requirement to emulators?
Are stronger bounds possible than the ones known for spanners?  Making progress on this requires answers to two related questions:
\begin{enumerate}
    \item How should we even \emph{define} a fault-tolerant emulator?  As we discuss shortly, there are two different definitions that both seem plausible at first glance.
    
    \item The lower bound on VFT spanners of~\cite{BDPW18} can also be generalized into an incompressibility argument, like the one by Thorup and Zwick~\cite{thorup2005approximate}.
    Since an emulator is just a different way of compressing distances, why wouldn't this lower bound apply to fault-tolerant emulators, ruling out hope for a better size/stretch tradeoff?
\end{enumerate}

These questions turn out to have some surprising answers.
We first argue that, of the two \emph{a priori} reasonable definitions of fault-tolerant emulators, only one of them is actually sensible.  We then show that this definition escapes the incompressibility lower bound, and we design fault-tolerant emulators that are sparser than the known \emph{lower bounds} for fault-tolerant spanners by $\poly(f)$ factors.
We also discuss fault-tolerant emulators with \emph{additive} stretch, and show that these also enjoy substantial improvements in size/stretch tradeoff over fault-tolerant additive spanners.

\subsubsection{VFT Emulator Definition}

Before we can even discuss bounds or constructions, we need to define fault-tolerant emulators.
Following Definition~\ref{def:ftspanner}, we get the following definition:
\begin{definition} [VFT Emulator Template Definition] \label{def:ftems}
Let $G = (V, E, w)$ be a weighted graph.
A graph $H$ is an $f$-vertex fault-tolerant ($f$-VFT) $t$-emulator of $G$ if, for all $F \subseteq V$ with $|F| \le f$, $H \setminus F$ is a $t$-emulator of $G \setminus F$.
\end{definition}

However, there are two reasonable definitions of (non-faulty) $t$-emulators that we could plug into this template definition.
These definitions are functionally equivalent in the non-faulty setting, but they give rise to two importantly different definitions of VFT emulators.

\begin{enumerate}
    \item One natural possibility is to define a weighted graph $H$ to be a $t$-emulator of $G$ if it satisfies
    $$\dist_G(u,v) \leq \dist_H(u,v) \leq t \cdot \dist_H(u,v)$$
    for all nodes $u, v$.
    Plugging this into Definition~\ref{def:ftems}, we get that a weighted graph $H$ is an $f$-VFT emulator of $G$ if, for any fault set $F \subseteq V, |F| \le f$ and vertices $u, v \in V \setminus F$, we have
    $$\dist_{G \setminus F}(u,v) \leq \dist_{H \setminus F}(u,v) \leq t \cdot \dist_{G \setminus F}(u,v).$$

    \item Recall that in the non-faulty setting, we always set the weight of an emulator edge $\{u,v\}$ to be exactly $w(u,v) = \dist_G(u,v)$: we need $w(u,v) \geq \dist_G(u,v)$ in order to ensure that $\dist_G(u,v) \leq \dist_H(u,v)$, and there is no benefit to setting $w(u,v) > \dist_G(u,v)$.  
    In other words, we can define an emulator $H$ as an \emph{unweighted} graph, where the weight of each edge $\{u, v\}$ simply \emph{becomes} the corresponding distance $\dist_G(u,v)$ in the input graph.
    We then say that $H$ is a $t$-emulator if it satisfies the usual distance inequalities
    $$\dist_G(u,v) \leq \dist_H(u,v) \leq t \cdot \dist_G(u,v)$$
    after this reweighting.
    This is a subtle distinction, since there is no important difference from the previous one in the non-faulty setting.
    But passed through Definition \ref{def:ftems}, it gives an importantly different definition of VFT emulators:

\begin{definition}[VFT Emulators] \label{def:ftemsintro}
Let $G = (V, E, w)$ be a weighted graph, and let $H$ be an unweighted graph on vertex set $V$.  For every fault set $F \subseteq V$ with $|F| \leq f$, for every $u,v \not\in F$ with $(u,v) \in E(H)$, we define weight function $w_F$ where $w_F(u,v) = \dist_{G \setminus F}(u,v)$. 

We then define $\dist_{H \setminus F}(u,v)$ to be the $u \leadsto v$ shortest path distance in $H \setminus F$ under weight function $w_F$.  We say that $H$ is an $f$-VFT $t$-emulator if
\[
\dist_{G \setminus F}(u,v) \leq \dist_{H \setminus F}(u,v) \leq t \cdot \dist_{G \setminus F}(u,v)
\]
for all $u,v \in V$ and for all $F \subseteq V$ with $|F| \leq f$ and $u,v\not\in F$.
\end{definition}
   
    In other words, for emulator edges in $H$, the edge weight in the post-fault graph $H \setminus F$ \emph{automatically updates} to be equal to the shortest-path distance between the endpoints in the remaining graph $G \setminus F$.
\end{enumerate}

Our next task is to point out that the second definition is the natural one to study, both mathematically and because it captures applications of fault-tolerant emulators in distributed systems.
Going forward, Definition \ref{def:ftemsintro} is the one we use.

\subsubsection{Theoretical Motivation for the Second Definition}

Although the first definition of VFT emulators may seem simpler, there is a pitfall when one attempts any construction under this definition.
Imagine that we add an edge $(u,v)$ to an emulator $H$, where $(u, v)$ is not also an edge in $G$.
Suppose we set its weight to $w(u, v) = \dist_G(u,v)$.
Then after any set of vertex faults that stretches $\dist_G(u, v)$ at all, the $u \leadsto v$ distance will be smaller in $H$ than in $G$, violating the lower distance inequality!
In general, one would always have to set emulator edge weights to be at least the \emph{maximum distance} $\dist_{G \setminus F}(u, v)$ over all possible vertex fault sets $F$.
This is an unnatural constraint, and it precludes most reasonable uses of emulator edges.
For example, if $G$ is a path with three nodes $u-x-v$ and we create an emulator edge $(u,v)$, then if $F = \{x\}$ we will have $\dist_{G \setminus F}(u,v) = \infty$.
Thus we are forced to set emulator edge weight $w(u, v) = \infty$, essentially disallowing this as an emulator edge at all.

The other issue is the incompressibility lower bounds from \cite{BDPW18}.
The lower bound on VFT spanners from~\cite{BDPW18} actually holds for all compression schemes: one cannot generally build a data structure on $o(f^{1-1/k} n^{1+1/k})$ bits that can report $(2k-1)$-approximate distances between all pairs of vertices under at most $f$ vertex faults.
The first definition of VFT emulators functions as such a compression scheme, so it cannot achieve improved bounds.

Why can we hope for the second definition of VFT emulators to \emph{bypass} this lower bound?
The answer lies in the fact that our emulator definition updates its edge weights under faults.
A VFT emulator \emph{cannot} actually be represented by a data structure of size approximately equal to the number of edges in the emulator, since a static data structure would not have this updating behavior.
In other words, since we assume that weight updates occur \emph{automatically}, we are not charging ourselves for the extra information one would have to carrry around in order to actually compute these updates.
This means it is \emph{a priori} possible that the second definition of fault-tolerant emulators can be significantly sparser than fault-tolerant spanners.

\subsubsection{Practical Motivation for the Second Definition.}

Now we explain the practical motivation behind the second definition.
Automatically updating edge weights may seem at first like an incredibly strong and unrealistic assumption.
Indeed, in some of the contexts in which spanners and emulators are used this would not be reasonable, e.g., as a preprocessing step for computing shortest paths~\cite{aingworth1999fast, dor2000all}. 
But spanners were originally designed for use in distributed computing~\cite{PelegU:89,PelegS:89}, and in distributed contexts, emulator edges typically represent \emph{logical} links rather than physical links.  That is, each emulator edge is treated as if it represents a path between the endpoints, since that is how packets/messages would actually travel between the endpoints.  An example of this is \emph{overlay networks}, where one builds a logical network that lives ``on top of'' another network (usually the Internet).
Overlay networks are extremely useful (even though they simply run on top of the Internet), and have been extensively studied, often either directly or indirectly using spanners, emulators, or related objects (e.g., \cite{BWDA17,RON,ADGHT06,Detour,OverQos}).

In a logical link on top of an underlying network, packets are automatically rerouted post-failures using some routing protocol on the underlying topology.
The vast majority of these routing protocols use shortest paths.  So for a logical link $(u,v)$, we would actually expect its distance to ``automatically" become $\dist_{G \setminus F}(u,v)$, where the seeming ``magic" of the edge weight update is implemented by the underlying routing algorithm converging on new shortest paths.  

So in applications of emulators to distributed computing, edges that take on weight equal to the remaining shortest path length is a very reasonable assumption.  Note that this does not obviate the need for emulators.  In an overlay network there will typically be a layer of routing in the overlay network itself in order to implement application-specific techniques and protocols (see, e.g., \cite{BWDA17}), so packets sent from $u$ to $v$ will follow shortest paths between $u$ and $v$ in the overlay.  Thus, these packets will experience stretch according to the stretch of the weight-updating emulator.

\subsection{Our Results}

Our previous discussion explains why it is \emph{possible} for VFT emulators to improve on the size/stretch tradeoff available to VFT spanners.
Our main results confirm this possibility; we construct VFT emulators that polynomially surpass the lower bounds for VFT spanners. 

\subsubsection{Multiplicative Stretch}
Our most general results (and main technical contributions) are in the multiplicative stretch setting, where we prove the following upper bound.
\begin{theorem} \label{thm:UB-main}
For all $k \geq 1$ and $f \leq n$, every $n$-node weighted graph $G=(V, E, w)$ admits an $f$-VFT $(2k-1)$-emulator $H$ with 
\[|E(H)| \leq \begin{cases} \Oish_k \left(f^{\frac12 - \frac{1}{2k}} n^{1+1/k} +fn \right)& \text{if $k$ odd} \\ \Oish_k \left(f^{\frac12} n^{1+1/k} +fn \right) & \text{if $k$ even.} \end{cases}
\]
Moreover, there is a randomized polynomial-time algorithm which constructs such an emulator with high probability.
\end{theorem}
In the above theorem, $\Oish_k$ hides factors that are polylogarithmic in $n$, and also factors that are exponential in $k$.\footnote{When $k$ is super-constant, \cite{BP19, DR20, BDR21} already give an upper bound of $O(fn^{1+o(1)})$ for VFT spanners, which cannot be improved beyond $O(fn)$ even with emulators.  Thus, the size gaps are already subpolynomial except in the parameter regime where $k = O(1)$.}
We typically think of $f$ as being polynomial in $n$, and in this setting (and when $k$ is a constant at least $3$), our emulators improve polynomially on VFT spanners.

The algorithm we design to prove Theorem~\ref{thm:UB-main} starts from the basic greedy VFT spanner algorithm of~\cite{BDPW18} (and its polytime extension in~\cite{DR20}), where we consider edges in nondecreasing weight order and add an edge if there is a fault set that forces us to add it.  To take advantage of the power of emulators, though, we augment this with an extra ``path sampling'' step: intuitively, when we decide to add a spanner edge, we also flip a biased coin for every $k$-path that it completes to decide whether to also add an emulator edge between the endpoints of the path.
(By ``complete a path,'' we mean that all previous edges in the path were already in the spanner and this is the last edge in the path to be added; the edge is not necessarily the first or last one in the path.)
These extra emulator edges do not \emph{replace} the added spanner edge, i.e., we do not add the emulator edge \emph{instead} of the spanner edge.
Instead, they act to help protect \emph{future} graph edges in the ordering, making it less likely that we will need to add spanner edges downstream.

Our two main lemmas roughly form a counting argument over the $k$-paths of the graph.
They are 1) with high probability there are not too many $k$-paths in our final emulator, meaning that we probably don't sample too many emulator edges as we go, and 2) if the emulator has many spanner edges then it has many $k$-paths, so we can't have added too many spanner edges either.  Combining these with an appropriate parameter balance gives us Theorem~\ref{thm:UB-main}.

The technical details of this analysis get surprisingly tricky, and it turns out that we actually cannot consider \emph{all} $k$-paths in the algorithm and analysis outlined above, but only a carefully selected subset of them that we call ``SALAD'' paths. 
There are several entirely different technical challenges absorbed into this definition, but the main one is the following.
Focusing momentarily on $k=3$, the core of our argument is that for most added spanner edges $(u, v)$, there are \emph{many distinct node pairs} $(x, y)$ for which $(u, v)$ completes an $x \leadsto y$ $3$-path.
The natural counting arguments give the best bounds by specifically considering $3$-paths that use $(u, v)$ as their middle edge, i.e., have the form $(x, u, v, y)$.
This requires that $(u, v)$ is heavier than both $(x, u)$ and $(v, y)$, so that it is considered last in the greedy algorithm; this motivates a focus on ``middle-heavy'' $3$-paths rather than arbitrary $3$-paths.
In fact, the cost/benefit of sampling emulator edges associated to non-middle-heavy $3$-paths is not worth it, so it is important to \emph{only} consider middle-heavy $3$-paths in our algorithm.

Generalizing the middle-heavy property to larger $k$ is nontrivial.
We use two properties that we call ``alternating'' (which also appeared in \cite{BDR22}) and ``dispersed'' (which is new here); these give the AD in the SALAD acronym.
The SAL part of the acronym stands for \emph{simple, avoids faults, local}, and we overview and motivate these at the beginning of Section \ref{sec:general-k}.
The full technical details of SALAD paths are responsible for the $\exp(k)$ factors and the even/odd $k$ distinction (a similar even/odd distinction appears in \cite{BDR22}, for a similar technical reason).

We complement our emulator upper bound with a nearly matching lower bound, which is a relatively straightforward extension of the \emph{edge}-fault-tolerant lower bound for spanners from~\cite{BDPW18}.  The case of $k=2$ (stretch $3$) is slightly different, so we handle it separately.

\begin{restatable}{theorem}{LBsmall} \label{thm:LB-3}
For all positive integers $n, f$ with $f \leq n$, there exists an unweighted $n$-node graph with $\Omega(f^{1/2} n^{3/2})$ edges for which any $f$-VFT $3$-emulator must have at least $\Omega(f^{1/2} n^{3/2})$ edges.
\end{restatable}

\begin{restatable}{theorem}{LBmain} \label{thm:LB-main}
Assuming the Erd\H{o}s girth conjecture, for all $k \geq 3$ and $f \leq n$ there is an unweighted $n$-node graph in which every $f$-VFT $(2k-1)$-emulator has at least $\Omega(k^{-1} f^{\frac12 - \frac{1}{2k}} n^{1+1/k})$ edges. 
\end{restatable}
This lower bound matches our upper bound for constant odd $k$, and is off by only an $f^{1/(2k)}$ factor for constant even $k$. 
An easy folklore observation implies that any $f$-regular input graph requires $\Omega(fn)$ edges for an $f$-VFT emulator, so our $+fn$ terms cannot be removed either.

\subsubsection{Additive Stretch}
Spanners and emulators are also studied in the context of \emph{additive} stretch: a \emph{$+k$-spanner/emulator} $H$ of an input graph $G$ is one that satisfies the distance inequality
$$\dist_G(u, v) \le \dist_H(u, v) \le \dist_G(u, v) + k$$
for all nodes $u, v$.
We have a complete understanding of the possibilities for additive emulators.
It is known that every unweighted input graph has a $+2$-emulator on $O(n^{3/2})$ edges \cite{aingworth1999fast} and a $+4$-emulator on $O(n^{4/3})$ edges \cite{dor2000all}.
These emulators are optimal, both in the sense that neither size bound can be improved at all, and in the sense that no $+c$-emulator can achieve $O(n^{4/3 - \eps})$ edges, even when $c$ is an arbitrarily large constant \cite{abboud20174}.
For spanners, our understanding lags only slightly behind: all graphs have $+2$-spanners on $O(n^{3/2})$ edges \cite{aingworth1999fast, knudsen2014additive}, $+4$-spanners on $\Oish(n^{7/5})$ edges \cite{chechik2013new}, and $+6$-spanners on $O(n^{4/3})$ edges \cite{baswana2010additive, knudsen2014additive,woodruff2010additive}.

Braunschvig, Chechik, Peleg, and Sealfon \cite{BGPS15} were the first to introduce fault-tolerance to additive spanners, via the natural extension of Definition \ref{def:ftemsintro}.

Unfortunately, it turns out that the price of fault-tolerance for additive spanners with fixed error is untenably high.
It is proved in \cite{BGPW17} that, for any fixed constant $c$, there are graphs on which an $f$-VFT $+c$ -spanner needs $n^{2 - \Omega(1/f)}$ edges.
In other words, tolerating \emph{one} additional fault costs $\poly(n)$ in spanner size, and there is no way to tolerate $\Omega(\log n)$ faults in subquadratic size.
Accordingly, constructions of VFT spanners of fixed size have to pay super-constant additive error of type $+O(f)$ \cite{BGPS15,bilo2015improved,parter2017vertex, BGPW17}.

We define VFT additive emulators with similar weight-updating behavior as in the multiplicative setting, with the same motivation.
We then show that these emulators actually avoid the undesirable size/fault-tolerance tradeoff suffered by VFT spanners.
We show the following extensions of the $+2$ and $+4$ emulators:

\begin{restatable}{theorem}{addtwo} \label{thm:additive-2}
For all $f \leq n$, every $n$-node unweighted graph $G=(V, E)$ admits an $f$-VFT $+2$-emulator $H$ with $|E(H)| = \Oish(f^{1/2}n^{3/2})$ edges.  There is also a randomized polynomial-time algorithm which computes such an emulator with high probability.
\end{restatable}

\begin{restatable}{theorem}{addfour} \label{thm:additive-4}
For all $f \leq n$, every $n$-node unweighted graph $G=(V, E)$ admits an $f$-VFT $+4$-emulator $H$ with $|E(H)| = \Oish\left(f^{1/3} n^{4/3}+nf \right)$ edges.  There is also a randomized polynomial-time algorithm which computes such an emulator with high probability.
\end{restatable}

The main point of these results is that the price of fault-tolerance for additive emulators is a \emph{multiplicative factor depending only on $f$}, rather than the parameter $f$ appearing in the exponent of the dependence on $n$, as it does for VFT additive spanners.
Moreover, the $f$-factors we obtain are essentially tight by our previous lower bound.  Any $+2$-emulator is also a $3$-emulator and hence by Theorem~\ref{thm:LB-3} must have size at least $\Omega(f^{1/2} n^{3/2})$, and any $+4$-emulator is also a $5$-emulator and so by Theorem~\ref{thm:LB-main} must have size at least $\Omega(f^{1/3} n^{4/3})$.

\subsection{Outline}
We begin by giving an \emph{exponential-time} algorithm for Theorem \ref{thm:UB-main} in the special case $k=3$ in Section~\ref{sec:k=3}.  This introduces the main ideas and approach that we use to prove Theorem \ref{thm:UB-main} in general, but it also happens to avoid a few technical details that become necessary only when we move to larger $k$ (allowing us to replace the complicated SALAD paths with simpler ``middle-heavy fault-avoiding'' paths).
We then prove the edge bound for Theorem~\ref{thm:UB-main} in its full generality: in Section~\ref{sec:general-k} we design a (still exponential-time) algorithm which proves \emph{existence} of sparse fault-tolerant emulators for all $k$.
Finally, in Section~\ref{sec:polytime} we show how to use ideas from~\cite{DR20} to make the algorithm polynomial-time without significant loss in emulator sparsity.  We then prove our lower bounds (Theorems \ref{thm:LB-3} and \ref{thm:LB-main}) in Section~\ref{sec:LB}, and we conclude with our results on additive spanners in Section~\ref{sec:additive}.

\section{Warmup: $k=3$ (Stretch $5$)} \label{sec:k=3}

We will warm up by proving the following special case of Theorem \ref{thm:UB-main}:
\begin{theorem} [Theorem \ref{thm:UB-main}, $k=3$]
For all $f \le n$, every $n$-node weighted graph $G = (V, E, w)$ has an $f$-VFT $5$-emulator $H$ with
$|E(H)| \le \Oish\left(f^{1/3} n^{4/3}\right) + O(fn)$.
\end{theorem}

Our exponential-time algorithm for $5$-emulators implementing this theorem is given in Algorithm~\ref{alg:k=3}.
We incrementally build an emulator $H$ by starting with an empty graph and adding edges.
We designate our added edges as \emph{spanner edges} (which are always contained in the input graph) and \emph{emulator edges} (which are not generally in the input graph).
We then let $H^{(sp)}$ be the subgraph of $H$ containing only its spanner edges, and let $H^{(em)}$ be the subgraph of $H$ containing only its emulator edges.

The algorithm is defined with respect to a parameter $d$.
Intuitively we can think of $d$ as (roughly) the desired average node degree in our final emulator: we will set $d= f^{1/3}n^{1/3}$.

\FloatBarrier
\begin{algorithm} [h!]
\textbf{Input:} Graph $G = (V, E, w)$, positive integer $f$\;
~\\

Let $H \gets (V, \emptyset, w)$ be the initially-empty emulator\;
\ForEach{edge $(u, v) \in E$ in order of nondecreasing weight $w(u, v)$}{
    \If{there is $F \subseteq V \setminus \{u, v\}$ of size $|F| \le f$ with $\dist_{H \setminus F}(u, v) > 5 \cdot w(u, v)$}{
        Add $(u, v)$ as a spanner edge to $H$\;
        \ForEach{$s \in N_{H^{(sp)} \setminus F}(u)$ and $t \in N_{H^{(sp)} \setminus F}(v)$}{
           Add $(s, t)$ as an emulator edge with probability $d^{-2}$\;
        }
    }
}
\Return{$H$};
\caption{Algorithm for $f$-VFT $5$-emulators}
\label{alg:k=3}
\end{algorithm}
\FloatBarrier

We begin by proving correctness.

\begin{lemma} \label{lem:k=3-correctness}
The graph $H$ returned by Algorithm~\ref{alg:k=3} is an $f$-VFT $5$-emulator of $G$.
\end{lemma}
\begin{proof}
It is easy to see (and essentially standard) that we just need to prove that $\dist_{H \setminus F}(u,v) \leq 5 \cdot w(u,v)$ for each edge $(u, v) \in E$ and possible fault set $F$:
by considering shortest paths in $G \setminus F$, this suffices to imply that $\dist_{H \setminus F}(x,y) \leq 5 \cdot \dist_{G \setminus F}(x,y)$ for all $x,y,F$ with $x,y\not\in F$, and hence implies that $H$ is an $f$-VFT $5$-emulator of $G$.

So let $(u,v) \in E$, and let $F \subseteq V \setminus \{u,v\}$ with $|F| \leq f$.  If $(u,v) \in E(H)$, then this is trivially true since then $\dist_{H \setminus F}(u,v)  \leq w(u,v) = \dist_{G \setminus F}(u,v)$.  Otherwise, Algorithm~\ref{alg:k=3} did not add $(u,v)$ to $H$.  By the condition of the if statement, this implies that $\dist_{H \setminus F}(u,v) \leq 5 \cdot w(u,v)$ as claimed.
\end{proof}

We now move to the more difficult (and interesting) task of proving sparsity.
We assume for convenience that all edge weights in the input graph $G$ are unique, so that we may unambiguously refer to \emph{the heaviest} edge among a set of edges.
If not, the following argument still goes through if we break ties between edge weights by the order in which the edges are considered by the algorithm.
We need to bound the number of spanner edges \emph{and} the number of emulator edges in the construction; our strategy is to count the number of instances of a particular structure in $H^{(sp)}$ called \emph{middle-heavy fault-avoiding $3$-paths}, and then we will use this counting in two different ways to bound the number of spanner edges in $H$, and then the number of emulator edges in $H$.

\subsection{Sparsity Analysis}
We start with the definitions of the paths that we care about, and then prove some of their properties and begin to count them.

\begin{definition} [Middle-Heavy $3$-Paths]
A $3$-path $\pi$ with node sequence $(s, u, v, t)$ is \emph{middle-heavy} if its middle edge is its heaviest one; that is, $w(u, v) > w(s, u)$ and $w(u, v) > w(v, t)$.
\end{definition}
When the edge $(u,v)$ is added to a middle-heavy path $\pi$ in $H^{(sp)}$, we say that $\pi$ is \textit{completed by} $(u,v)$ (i.e.~after adding $(u,v)$ $\pi$ exists in $H^{(sp)}$).

For every edge $(u,v)$ added by the algorithm, there must exist some set $F_{(u,v)}$ with $|F_{(u,v)}| \leq f$ such that $\dist_{H \setminus F_{(u, v)}}(u, v) > 5 \cdot w(u, v)$ (or else the algorithm would not have added $(u,v)$).
The algorithm uses this set in order to determine the neighborhoods $N_{H^{(sp)} \setminus F}(u),N_{H^{(sp)} \setminus F}(v)$ from which it samples its emulator edges.
If there are multiple such sets then the choice is arbitrary.
In the following, $F_{(u, v)}$ refers to the specific set selected by the algorithm when considering the edge $(u, v)$.

\begin{definition} [Fault-Avoiding Paths]
A path $\pi$ in $H^{(sp)}$ with heaviest edge $(u, v)$ is \emph{fault-avoiding} if $\pi \cap F_{(u, v)} = \emptyset$.
\end{definition}

We first prove an auxiliary counting lemma.
Let $C_{(s, t)}$ count the number of middle-heavy fault-avoiding $3$-paths from $s$ to $t$ in $H^{(sp)}$ at a given moment in the algorithm.
Whenever we choose to add a spanner edge $(u, v)$, we define the set
$$\Psi(u, v) := \left\{ (s, t) \in V \times V \ \mid \ (u, v) \text{ completes a middle-heavy fault-avoiding $3$-path from $s$ to $t$} \right\}.$$
The following lemma gives a certain kind of control on the values that $C_{(s, t)}$ can reach:

\begin{restatable}{lemma}{countermass} \label{lem:countermass}
With high probability, whenever we add a new spanner edge $(u, v)$ in our algorithm, we have
$$\sum \limits_{(s, t) \in \Psi(u, v)} C_{(s, t)} \le \Oish(fd^2)$$
where the values $C_{(s, t)}$ are defined just \emph{before} $(u, v)$ is added to $H$.
\end{restatable}
\begin{proof}[Proof Sketch.]
We defer the full proof to Appendix~\ref{app:proofs}, since the details are technical and do not provide much additional insight.  Intuitively, this lemma is true because the counter value $C_{(s,t)}$ corresponds to the number of different times we flipped a coin to decide whether or not to add $(s,t)$ as an edge (since $C_{(s,t)}$ is the number of middle-heavy $3$-paths between $s$ and $t$, and for each such path we flip such a coin).  Since each coin has bias $d^{-2}$ by the definition of the algorithm, if $\sum_{(s, t) \in \Psi(u, v)} C_{(s, t)}$ were much larger than $d^2$ then with high probability there would already be an emulator edge $(s,t)$ where $(s,t) \in \Psi(u,v)$.  And if such an edge existed, the path $u-s-t-v$ would have stretch at most $5$, and hence we would not have added $(u,v)$.  

Making this formal requires union bounding over all possible fault sets $F$ in the definition of fault-avoiding rather than just considering $F_{(u,v)}$, which also causes the extra factor of $f$ in the lemma statement.  This introduces significant extra notation but is a straightforward calculation, so we defer it to Appendix~\ref{app:proofs}. 
\end{proof}

We can now use Lemma~\ref{lem:countermass} to bound the number of middle-heavy fault-avoiding 3-paths.

\begin{lemma} \label{lem:mfpathcount}
With high probability, there are $d^2 |E(H^{(sp)})| + \Oish\left(fn^2 \right)$ total middle-heavy fault-avoiding $3$-paths in the final graph $H^{(sp)}$.
\end{lemma}
\begin{proof}
For each edge $(u, v)$ added to the emulator, let us split into two cases by the size of $\Psi(u, v)$.
Notice that, since a middle-heavy fault-avoiding path completed by $(u, v)$ is uniquely determined by $(u, v)$ and its endpoints, the size of $\Psi(u, v)$ is the same as the number of middle-heavy fault-avoiding paths completed by $(u, v)$.

\paragraph{Case 1: $|\Psi(u, v)| \le d^2$.}
In this case, the edge $(u, v)$ completes $\le d^2$ new middle-heavy fault-avoiding $3$-paths.
By a unioning, only $d^2 |E(H^{(sp)}|$ middle-heavy fault-avoiding $3$-paths can be completed by edges of this type.

\paragraph{Case 2: $|\Psi(u, v)| > d^2$.}
Assuming the high-probability event from Lemma \ref{lem:countermass} holds, we also have
$$\sum \limits_{(s, t) \in \Psi(u, v)} C_{(s, t)} = \Oish\left(f d^2 \right).$$
Thus, the average value of $C_{(s, t)}$ over the $>d^2$ node pairs in $\Psi(u, v)$ is $\Oish(f)$.
So by Markov's inequality, for at least half of the node pairs $(s, t) \in \Psi(u, v)$, we have $C_{(s, t)} = \Oish(f)$.

This implies that only $\Oish(fn^2)$ middle-heavy fault-avoiding $3$-paths may be completed by edges from this case, by a straightforward amortization argument over all pairs $(s,t)$.
Whenever a middle-heavy fault-avoiding $3$-path $\pi = (s, u, v, t)$ is completed by an edge in this second case, let us say that $\pi$ is \emph{dispersed} if $C_{(s, t)} = \Oish(f)$.

By the previous argument, at least half of all paths completed by edges in this case are dispersed, so it suffices to only count the dispersed paths.
Moreover, by definition of $C_{(s, t)}$ every dispersed path from $s$ to $t$ is among the first $\Oish(f)$ middle-heavy $3$-paths from $s$ to $t$; thus, unioning over all choices of $s, t$ there are $\Oish(fn^2)$ dispersed paths in total.

Combining the two cases, we get at most $d^2 |E(H^{(sp)}| + \Oish(fn^2)$ middle-heavy fault-avoiding $3$-paths in $H^{(sp)}$.
\end{proof}

 We now show how to use the above bound on middle-heavy fault-avoiding $3$-paths to bound the number of edges in our emulator.  We first bound the number of emulator edges (edges which were added by the path sampling and so might not be in $E$) in terms of the number of spanner edges (edges from $E$), and then bound the number of spanner edges.

\begin{lemma} [Emulator Edge to Path Counting] \label{lem:3emedge}
With high probability, we have
$$\left|H^{(em)}\right| \le O\left(\left|E\left(H^{(sp)}\right)\right|\right) + \Oish\left(\frac{fn^2}{d^2}\right).$$
\end{lemma}
\begin{proof}
Let $\alpha = d^2 |E(H^{(sp)}| + \Oish(fn^2)$ be the bound on the number of middle-heavy fault-avoiding 3-paths in $H^{(sp)}$ which holds with high probability from Lemma~\ref{lem:mfpathcount}.  Consider the following two events.
\begin{itemize}
    \item Let $\mathcal A$ be the event that $H^{(sp)}$ has at most $\alpha$ middle-heavy fault-avoiding paths.  We know from Lemma~\ref{lem:mfpathcount} that this holds with high probability.
    
    \item Whenever the algorithm considers adding some emulator edge, we call this an \emph{attempt}.  Let $X_i$ be an indicator random variable for the event that the $i^{th}$ attempt is successful (meaning that the emulator edge is actually added).
    If there is no $i^{th}$ attempt since the algorithm has terminated before $i$ attempts are made, then we set $X_i = 1$ with probability $d^{-2}$ and $X_i = 0$ with probability $1-d^{-2}$.  Note that $\E[X_i] = d^{-2}$ for all $i$.  Moreover, note that $X_i$ and $X_j$ are independent for $i \neq j$. 
    Let $X = \sum_{i=1}^{\alpha} X_i$, and let $\mathcal B$ be the event that $X < 2d^{-2} \alpha$.  So if $\mathcal B$ occurs, then of the first $\alpha$ attempts, at most $2d^{-2} \alpha$ emulator edges are added.  
    By linearity of expectations we know that $\E[X] = d^{-2} \alpha$.  Moreover, we know that the $X_i$'s are independent and that $d^{-2} \alpha = \omega(\log n)$.  Hence a standard Chernoff bound implies that $\mathcal B$ occurs with high probability.  
\end{itemize}

Since both $\mathcal A$ and $\mathcal B$ occur with high probability, a simple union bound implies that $\mathcal A \cap \mathcal B$ occurs with high probability.  Note that every emulator edge is caused by some attempt, and that the number of attempts is precisely equal to the number of middle-heavy fault-avoiding 3-paths.  Hence if both $\mathcal A$ and $\mathcal B$ occur, the number of emulator edges in $H$ is at most $O(d^{-2} \alpha)$, as claimed. 
\end{proof}

\begin{lemma} [Spanner Edge to Path Counting] \label{lem:3spanedge}
Letting $\mu$ be the number of middle-heavy fault-avoiding $3$-paths in $H^{(sp)}$, we have
$\left| E\left(H^{(sp)}\right) \right| \leq O\left( \mu^{1/3} n^{2/3} + nf\right)$.
\end{lemma}
\begin{proof}
Let $c>0$ be a sufficiently small absolute constant and let $\delta$ be the average degree in $H^{(sp)}$.
If $f > c\delta$ then we have $O(nf)$ edges in $H^{(sp)}$, and we are done.
So assume in the following that $f \le c \delta$.

From here, we will use a slight variant on the ``subsampling'' method from extremal graph theory.
Starting with $H^{(sp)}$, we will define a random subgraph $H''$.
The definition of $H''$ is such that we can straightforwardly relate the number of middle-heavy fault-avoiding $3$-paths $\mu$ in $H^{(sp)}$ to the expected number of simple (does not repeat nodes) middle-heavy fault-avoiding $3$-paths $\mu''$ that survive in $H''$.
Separately, we will argue that enough edges probably remain in $H''$ to use a straightforward counting argument to lower bound the expectation of $\mu''$.
Together, these two parts imply a bound on $\mu$ that can be rearranged into our desired lemma statement.

We now state the definition of $H''$.
First, let $H'$ be a random induced subgraph of $H^{(sp)}$ obtained by independently keeping each node with probability $(c\delta)^{-1}$.
Let us say that an edge $(u, v)$ in $H'$ is \emph{clean} if none of the nodes in its associated fault set $F_{(u, v)}$ survive in $H'$.
We define $H''$ as the subgraph of $H'$ that contains only its clean edges.

\paragraph{Lower Bound on $\mathbb{E}[\mu'']$.}

First, let us analyze the probability that a given edge $(u, v)$ in $H^{(sp)}$ survives in $H''$.
The probability that $(u, v)$ survives in $H'$ is exactly $(c\delta)^{-2}$ (the event that $u, v$ each survive).
Conditional on $(u, v)$ surviving in $H'$, it is clean iff none of the nodes in $F_{(u, v)}$ also survive.
Since $|F_{(u, v)}| \le f \le c\delta$, $(u, v)$ is clean with constant probability. 
So $(u, v)$ survives in $H''$ with probability $\Omega((c\delta)^{-2})$, which implies
\[
\mathbb{E}\left[ \left| E(H'') \right| \right] = \left|E(H^{(sp)})\right| \cdot \Omega\left((c\delta)^{-2}\right) = \Omega\left( n \delta^{-1} c^{-2} \right).
\]
Meanwhile, the expected number of nodes that survive in $H''$ is exactly
$\mathbb{E}\left[ \left| V(H'') \right| \right] = n \delta^{-1} c^{-1}$.
Let us imagine that we start with an initially-empty graph on the vertex set $V(H'')$, and we add the edges in $E(H'')$ one by one in order of increasing weight.
For each added edge $(u, v)$ that is the \emph{first} edge incident to one of its endpoints $u$ or $v$, this edge does not create any new middle-heavy $3$-paths.
There are at most $|V(H'')|$ such edges.
Any other edge $(u,v)$ creates at least one simple middle-heavy $3$-path in $H''$.
Specifically, the $3$-path $(s, u, v, t)$ in which it is the middle edge must be middle-heavy by the order in which we are adding the edges, and it is simple since if $s=t$ then we are forced to include $s \in F_{(u, v)}$, but then $s$ must not survive in $G'$ (since $(u, v)$ is clean).
It follows that
\begin{align*}
\mathbb{E}[\mu''] &\ge \mathbb{E}\left[\left|E(H'')\right| - \left|V(H'')\right| \right] = \mathbb{E}\left[\left|E(H'')\right|\right] - \mathbb{E}\left[ \left|V(H'')\right| \right] = \Omega\left( n\delta^{-1} c^{-2}\right) - n\delta^{-1} c^{-1}\\ 
&= \Omega\left( n\delta^{-1} c^{-2}\right) \tag*{by choice of small enough $c>0$.}
\end{align*}

\paragraph{Upper Bound on $\mathbb{E}[\mu'']$.}

We can relate $\mu$ and $\mu''$ as follows.
We notice that every simple middle-heavy $3$-path $\pi$ in $H''$ must correspond to a \emph{fault-avoiding} $3$-path in $H^{(sp)}$.
This holds because if the middle edge $(u, v)$ of $\pi$ survives in $H''$, then it must be clean, implying that no nodes in $F_{(u, v)}$ survive in $H''$.

Now let $q$ be a middle-heavy fault-avoiding $3$-path in $H^{(sp)}$.
We notice that $q$ must be simple, since (as before) if $q = (s, u, v, s)$ then we would have to include $s \in F_{(u, v)}$ and so $q$ would not be fault-avoiding.
Since $q$ is simple it survives in $H'$ with probability exactly $(c\delta)^{-4}$, and thus it survives in $H''$ with probability $\le (c\delta)^{-4}$.
We therefore have
$\mathbb{E}[\mu''] \leq O\left( \mu (c\delta)^{-4} \right).$

\paragraph{Putting It Together.}
By the previous two parts, we have
$\Omega\left( n\delta^{-1} c^{-2}\right) \leq \mathbb{E}[\mu''] \leq O\left( \mu (c\delta)^{-4} \right)$,
which implies that $\delta = O((\mu/(nc^2))^{1/3})$, and thus
$n\delta = O_c\left( \mu^{1/3} n^{2/3} \right).$
Since $\delta$ is defined as the average degree in $H^{(sp)}$, this proves the lemma.
\end{proof}

Our size bound now follows by directly combining our previous three lemmas; see Appendix \ref{app:proofs}. 

\begin{restatable}{lemma}{kthreealgebra}
The emulator $H$ returned by Algorithm \ref{alg:k=3} has
$|E(H)| = \Oish\left(f^{1/3} n^{4/3}\right) + O(fn)$ with high probability.
\end{restatable}

\section{Vertex Fault-Tolerant $(2k-1)$-Emulators} \label{sec:general-k}

Our goal in this section is to prove Theorem \ref{thm:UB-main}.  We start by defining several properties of certain desired paths that let us generalize the algorithm.

\subsection{SALAD Paths and Proof Overview}

We begin by explaining, at a high level, the relationship between this argument for general $k$ and the one given previously for $k=3$.
The core of our previous proof was a counting argument over middle-heavy fault-avoiding $3$-paths in $H^{(sp)}$.
The core of our general argument will be a counting argument over ``SALAD'' $k$-paths in $H^{(sp)}$.
SALAD is an acronym for \textbf Simple, \textbf Alternating, \textbf Local, \textbf Avoids faults, \textbf Dispersed.
We will explain these five properties and their role in the analysis momentarily, but first let us state our algorithm. This algorithm uses a notion of \emph{local} paths, which we define immediately after the algorithm and do not have an analog in our simpler $k=3$ case. We say that a path in $H^{(sp)}$ is \textit{completed by} an edge $(u,v)$ if the path exists in $H^{(sp)}$ and $(u,v)$ is the heaviest edge in the path (i.e., the path exists in $H^{(sp)}$ once $(u,v)$ has been added).

\FloatBarrier
\begin{algorithm} [h!]
\textbf{Input:} Graph $G = (V, E, w)$, positive integer $f$, odd positive integer $k$\;
~\\

Let $H \gets (V, \emptyset, w)$ be the initially-empty emulator\;
\ForEach{edge $(u, v) \in E$ in order of nondecreasing weight $w(u, v)$}{
    \If{there is $F \subseteq V \setminus \{u, v\}$ of size $|F| \le f$ with $\dist_{H \setminus F}(u, v) > (2k-1) \cdot w(u, v)$}{
        Add $(u, v)$ as a spanner edge to $H$\;
        \ForEach{local path $\pi$ in $H^{(sp)}$ with $j \le k$ edges completed by $(u, v)$}{
           Add the endpoints $(s, t)$ of $\pi$ as an emulator edge with probability $d^{-(j-1)}$\;
        }
    }
}
\Return{$H$};
\label{alg:general-k}
\caption{Algorithm for $f$-VFT $(2k-1)$-emulators}
\end{algorithm}
\FloatBarrier

Stretch analysis of this algorithm is essentially the same as Lemma \ref{lem:k=3-correctness}; we include it here for completeness.  

\begin{lemma} \label{lem:stretch-k}
The emulator $H$ returned by Algorithm~\ref{alg:general-k} is an $f$-VFT $(2k-1)$-emulator.
\end{lemma}
\begin{proof}
Consider some $(u,v) \in E$ and $F \subseteq V \setminus \{u,v\}$ with $|F| \leq f$.  If $(u,v) \in E(H)$ then clearly $\dist_{H \setminus F}(u,v) \leq w(u,v)$.  Otherwise, Algorithm~\ref{alg:general-k} did not add $(u,v)$ to $H$, and so by the ``if'' condition we know that $d_{H \setminus F}(u,v) \leq (2k-1) \cdot w(u,v)$ as required.
\end{proof}

We now define SALAD paths:

\begin{itemize}
    \item \textbf{Simple}: $\pi$ does not repeat nodes.
    We implicitly required simplicity in our previous $k=3$ proof, since (as used in Lemma \ref{lem:3spanedge}) a non-simple middle-heavy path of the form $(s, u, v, s)$ is not fault-avoiding.
    In our extension, it is more convenient to make the simplicity requirement explicit.
    
    \item \textbf{Alternating}: $\pi$ is alternating if every even-numbered edge in $\pi$ is heavier than the two adjacent odd-numbered edges.
    That is: if $\pi$ has edge sequence $(e_1, \dots, e_k)$, then for all $i$, we have $w(e_{2i}) > w(e_{2i-1})$ and $w(e_{2i}) > w(e_{2i+1})$.
    If $k$ is even, then $e_k$ only needs to be heavier than $e_{k-1}$.
    
    ``Alternating'' turns out to be the natural extension of ``middle-heavy'' to paths of length $k \ge 3$ (notice for $k=3$, alternating and middle-heavy are the same).
    Roughly, our analysis will involve ``splitting'' paths over their heaviest edge and recursively analyzing the subpath on either side.
    Like for $k=3$, this splitting process is most efficient when the heaviest edge is somewhere in the middle of the path (neither the first nor last one).
    An alternating path is one where the heaviest edge \emph{remains} somewhere in the middle at every step of the recursion, until finally the path decomposes into individual edges.
    In fact, this is not \emph{quite} true in the case where $k$ is even (due to the last edge), which is precisely why our bounds are a little worse for even $k$.
    
    \item \textbf{Local}: this is a new property that does not have an analog in our previous $k=3$ argument.
    Let $b = \Theta(kd)$ be a parameter (it will be more convenient to specify the implicit constant later in the analysis).
    For each node $v$, we put the edges incident to $v$ in $H^{(sp)}$ into \emph{buckets} $\{B^i_v\}$: the first $b$ edges incident to $v$ are in its first bucket $B_v^1$, the next $b$ edges are in the second bucket $B_v^2$, etc.
    We define $\pi$ to be \emph{local} if, for any three-node contiguous subpath $(x, y, z) \subseteq \pi$, the edges $(x, y), (y, z)$ belong to the same bucket for $y$. 
    
    Locality is necessary because we sample SALAD paths of \emph{all} lengths $j \le k$.
    Our proof strategy from $k=3$ works just fine to limit the emulator edges \emph{contributed by SALAD paths of length $k$}.
    But it does not help us limit the emulator edges contributed by SALAD paths of shorter length $j < k$.
    By including locality explicitly, we gain an easy way to limit this quantity, at the price of a little more complexity in some of the downstream proofs.
    
    \item \textbf{Avoids Faults}: This is a slightly more stringent property than ``fault-avoiding'' used previously.
    Whenever we add a spanner edge $(u, v)$, let $F_{(u, v)}$ be a choice of fault set that forces $(u, v)$, just like in the $k=3$ proof.
    We say that $\pi$ avoids faults if, for \emph{every} edge $(u, v) \in \pi$ (not just the heaviest one), we have $\pi \cap F_{(u, v)} = \emptyset$.
    
    \item \textbf{Dispersed}: this property showed up briefly in the $k=3$ case, but we were able to bury it in the technical details of Lemma \ref{lem:mfpathcount}.
    Here, we need to bring it more to the forefront of the analysis.
    We will say that $\pi$ is a SALA path if it satisfies the first four properties described previously.
    Among the SALA paths, we will declare them either \emph{concentrated} or \emph{dispersed} as follows, and we will only use the dispersed ones in our analysis:
    
    \begin{itemize}
        \item Notice that we can split $\pi$ into two (possibly empty) shorter SALA paths $\pi_1, \pi_2$ by removing its heaviest edge $(u, v)$.
        If either of $\pi_1, \pi_2$ is concentrated, then $\pi$ is concentrated as well.
        If $\pi_1, \pi_2$ are both dispersed, then we will say that $\pi$ is \emph{splittable}, and it may still be concentrated or dispersed according to the following point:
        
        \item Set a \emph{threshold parameter} $\tau = \Oish(f)$.
        For all $1 \le j \le k$, among the splittable $j$-paths between each pair of endpoints $(s, t)$, the first $\tau^{\left\lceil\frac{j-1}{2}\right\rceil}$ completed paths are \emph{dispersed} and the rest are \emph{concentrated}.
        If two $s \leadsto t$ splittable paths are completed by the same edge, and thus arise in $H^{(sp)}$ at the same time, then we pick an arbitrary order so that the ``first'' paths are unambiguous.
    \end{itemize}
\end{itemize}

The inclusion of locality among our properties actually significantly changes the structure of the proof.
Because we consider a more restricted kind of path, it gets much easier to control the number of emulator edges (this is the whole point of locality):
\begin{lemma}\label{lem:general_em_size}
With high probability, $\left|E(H^{(em)})\right| \le \Oish\left(\left| E\left(H^{(sp)}\right) \right| \right) \cdot O(k)^k$.
\end{lemma}
\begin{proof}
One generates a local $j$-path by picking an oriented edge to be the starting edge, and then repeatedly extending the path by choosing $1$ edge among the at most $b$ possible edges satisfying the locality constraint.
Hence there are at most $O\left(|E(H^{(sp)}| \cdot b^{j-1} \right)$ local paths in $H^{(sp)}$.

Each local $j$-path  completed by a spanner edge $(u,v)$ is independently sampled as an emulator edge with probability $d^{-(j-1)}$.  
Thus, the \emph{expected} number of emulator edges contributed by local $j$-paths is
$$O\left(\left| E\left(H^{(sp)}\right) \right| \cdot \left(\frac{b}{d}\right)^{j-1} \right) \le \left| E\left(H^{(sp)}\right) \right| \cdot O(k)^{k-1} .$$
Since the edges are sampled independently, by a standard Chernoff bound,
\[\Oish\left(\left| E\left(H^{(sp)}\right) \right|\right) \cdot O(k)^{k-1}.\]
The lemma now follows by unioning over all choices of $j \le k$.
\end{proof}

On the other hand, it gets much harder to bound the number of spanner edges in $H$.
We use the following main technical lemma:
\begin{lemma} [Counting Lemma] \label{lem:counting}
Let $c$ be a large enough absolute constant, suppose $H^{(sp)}$ has average degree $\delta \ge cdk$, and also suppose $d \ge cf$.
Then with high probability, $H^{(sp)}$ has at least $n d^k$ SALAD $k$-paths.
\end{lemma}

Before proving this lemma, we can do some simple algebra to show why it implies a bound on spanner edges:

\begin{lemma}\label{lem:final_size}
If we set parameter
\[ d := \begin{cases} \max\left\{ \text{polylog } n \cdot f^{\frac12 - \frac{1}{2k}} n^{1/k}, cf \right\} & \text{if $k$ odd} \\
\max\left\{ \text{polylog } n \cdot f^{\frac12} n^{1/k}, cf \right\} & \text{if $k$ even,} \end{cases}
\] 
with large enough polylogs, then with high probability, we have
\[ \left|E(H^{(sp)}) \right| \leq \begin{cases} \Oish_k\left(f^{\frac12 - \frac{1}{2k}} n^{1+1/k} +fn\right)& \text{if $k$ odd} \\ \Oish_k\left(f^{\frac12} n^{1+1/k} +fn\right) & \text{if $k$ even} \end{cases}
\]
\end{lemma}

\begin{proof} [Proof, assuming Lemma \ref{lem:counting}]
By definition of dispersion, for each node pair $(s, t)$, we can have only $\Oish(f)^{\left \lceil \frac{k-1}{2} \right\rceil}$ total $s \leadsto t$ SALAD $k$-paths, so there are $\le n^2 \cdot \Oish(f)^{\left\lceil \frac{k-1}{2} \right\rceil }$ SALAD $k$-paths in total. Therefore the number of these paths is at most $n^2 \cdot \tilde{O}  (f)^{\lceil \frac{{k-1}}{2} \rceil}= n^2\cdot \tilde{O}(f)^{\frac{k-1}{2}}$ when $k$ is odd. Based on definition of $d$, and by a choice of large enough polylog, this means that $H^{(sp)}$ has \emph{strictly less than} $n^2 \cdot \tilde{O}(f)^{\frac{k-1}{2}} < nd^k$ SALAD $k$-paths.

If $k$ is even, there are at most $n^2 \cdot \tilde{O} (f)^{\lceil \frac{k-1}{2} \rceil}= n^2\cdot \tilde{O}(f)^{\frac{k}{2}}$.
Similarly by choosing a large enough polylogarithmic factor in the definition of $d$ for the even case, we also have that the number of SALAD $k$-paths is strictly less than $nd^k$. 

In both cases, by applying the counting lemma in contrapositive, we conclude that the average degree in $H^{(sp)}$ is $\delta = O(dk)$.
Thus $H^{(sp)}$ has $O(n\delta)$ edges, and by plugging in $d$ the claim follows.
\end{proof}

And now it is trivial to prove Theorem~\ref{thm:UB-main}.
\begin{proof}[Proof of Theorem \ref{thm:UB-main}]
The combination of Lemma~\ref{lem:general_em_size} (which bounds the number of edges of $H^{(sp)}$) and Lemma~\ref{lem:final_size} (which relates the number of emulator edges added to $E(H^{(sp)})$) implies the theorem.
\end{proof}

So it just remains to prove our counting lemma, which is the main technical part of the proof.

\subsection{Counting Lemma}

Towards proving our counting lemma, our first task is to extend Lemma \ref{lem:countermass} from the $k=3$ case.
We will define slightly more expressive variables: let $C^j_{(s, t)}$ count the number of \emph{local} $s \leadsto t$  $j$-paths at a given moment in the algorithm.
We also define sets
$$\Psi^j(u, v) := \left\{(s, t) \in V \times V \ \mid \ (u, v) \text{ completes a SALA splittable $j$-path from $s$ to $t$}\right\}.$$
The following lemma controls the values that $C^j_{(s, t)}$ can reach:

\begin{lemma} \label{lem:countermassk}
With high probability, for all spanner edges $(u, v)$ added to $H$ and all $1 \le j \le k$, just before $(u, v)$ is added we have
\[\sum \limits_{(s, t) \in \Psi^j(u, v)} C^j_{(s, t)} = \Oish\left(f d^{j-1} \right).\]
\end{lemma}
\begin{proof}

The proof is similar to Lemma \ref{lem:countermass}.  Intuitively, if $\sum \limits_{(s, t) \in \Psi^j(u, v)} C^j_{(s, t)}$  is large enough, then with high probability there will already be an emulator edge $(s,t)$ for some $(s,t) \in \Psi^j(u,v)$, which would mean that we would not have actually added $(u,v)$ to $H$.  To formalize this, though, we need to analyze even edges that were not added to $H$ and take a union bound over all possible fault sets, as in Lemma~\ref{lem:countermass}. 

So we begin as in Lemma~\ref{lem:countermass}.  Let $(u,v)$ be an edge in the input graph, and let $F \subseteq V \setminus \{u,v\}$ with $|F| \leq f$.  Consider the moment in the algorithm where we inspect $(u, v)$ and decide whether or not to add it to the emulator (note: $(u, v), F$ are arbitrary; we may or may not actually add $(u, v)$, and if we do, we do not necessarily have $F = F_{(u, v)}$).
We use the following extensions of our previous definitions:
\begin{itemize}
    \item For a path $\pi$ in $H^{(sp)}$ that \emph{would} be completed, if we added $(u, v)$ to the emulator, we say that $\pi$ is \emph{$F$-avoiding} if $\pi \cap F = \emptyset$. 
    
    \item $\Psi^j(u,v,F)$ is the set of node pairs $(s, t) \in V \times V$ such that, if we added $(u, v)$ to the emulator, it would complete at least one new $F$-avoiding SALA $j$-path from $s$ to $t$.
    \item We say that $F$ is \emph{mass-avoiding} for $(u,v)$ and $j$ if
\[
\sum_{(s,t) \in \Psi(u,v,F)} C^j_{(s,t)} > c fd^{j-1} \log n.
\]
where $c$ is some large enough absolute constant.
\end{itemize}

Note that the lemma statement is equivalent to the claim that, \emph{if} $(u,v)$ is added to $H^{(sp)}$, \emph{then} $F_{(u,v)}$ is not mass-avoiding.
We have set things up for general $(u, v), F$ because our proof strategy is to take a union bound over \emph{all} possible choices of $(u, v), F$, which will thus include $F_{(u, v)}$.

We say that a mass-avoiding $F$ is \emph{good} for $(u,v)$ if (immediately prior to $(u,v)$ being considered by the algorithm) there is some  $(s,t) \in \Psi^j(u,v,F)$ such that $(s,t)$ is already an emulator edge in $H$.  Otherwise, we say that $F$ is \emph{bad} for $(u,v)$.

We now prove that with high probability, every mass-avoiding $F$ is good for $(u,v)$.  To see this, consider some mass-avoiding $F$.  Every local $j$-path $\pi$ which contributes to $C^j_{(s,t)}$ was completed by some (even) edge, so by definition of the algorithm, when $\pi$ was completed we sampled $(s,t)$ as an emulator edge with probability $d^{-(j-1)}$.  This is true for every such $s \leadsto t$ local $j$-path, so the algorithm independently adds $(s,t)$ as an emulator edge with probability $d^{-(j-1)}$ at least $C^j_{(s,t)}$ times.   These choices are also clearly independent for different pairs $(s,t)$ and $(s', t')$, and hence the probability that $F$ is bad is at most
\begin{align*}
\prod_{(s,t) \in \Psi^j(u,v,F)} \left(1-\frac{1}{d^{(j-1)}}\right)^{C^j_{(s,t)}} &
\leq \exp\left( - \frac{\sum_{(s,t) \in \Psi^j(u,v,F)} C^j_{(s,t)}}{d^{j-1}}\right) 
\leq \exp\left( -cf \log n\right) 
\leq 1/n^{f+10},
\end{align*}
where we used that $F$ is mass-avoiding and we set $c$ sufficiently large.  There are at most $n^f$ possible mass-avoiding sets $F$ (since $|F| \leq f$), so a union bound over all all of them implies that every mass-avoiding set $F$ is good for $(u,v)$ with probability at least $1-1/n^{10}$.  
We can now do another union bound over all $(u,v)$ to get that this holds for \emph{every} $(u,v) \in E$ (whether added to $H^{(sp)}$ or not) with probability at least $1-1/n^8$.  

Now consider some $(u,v) \in H^{(sp)}$.  By the above, if $F_{(u,v)}$ is mass-avoiding, then it must be good.  Hence there is some emulator edge $(s,t)$ with $(s,t) \in \Psi^j(u,v,F_{(u,v)})$, so by the definition of $\Psi(u,v,F_{(u,v)})$ there is some $s \leadsto t$ SALA $j$-path in $H^{(sp)}$ that is completed by $(u,v)$.  Note that by the definition of a SALA path, we know that no vertices in this path are in $F_{(u,v)}$.  Thus immediately prior to adding $(u,v)$ to $H^{(sp)}$, it was the case that 
\begin{align*}
\dist_{H \setminus F_{(u,v)}}(u,v) &\leq \dist_{H \setminus F_{(u,v)}}(u,s) + \dist_{H \setminus F_{(u,v)}}(s,t) + \dist_{H \setminus F_{(u,v)}}(t,v) \\
&\leq \dist_{G \setminus F_{(u,v)}}(u,s) + \dist_{G \setminus F_{(u,v)}}(s,t) + \dist_{G \setminus F_{(u,v)}}(t,v) \\
&\leq \dist_{G \setminus F_{(u,v)}}(u,s) + \left(\dist_{G \setminus F_{(u,v)}}(s,u) + \dist_{G \setminus F_{(u,v)}}(u,v) + \dist_{G \setminus F_{(u,v)}}(t,v) \right) + \dist_{G \setminus F_{(u,v)}}(t,v) \\
&\leq (j-1) \cdot \dist_{G \setminus F_{(u,v)}} (u,v) + j \cdot \dist_{G \setminus F_{(u,v)}} (u,v)\\
&= (2j-1) \cdot w(u,v) \leq (2k-1) \cdot w(u,v).
\end{align*}
In the above inequalities we used the triangle inequality, the fact that $(u,v)$ is the heaviest edge in the SALA $s\leadsto t$ $j$-path since edges are added in increasing weight order, and the fact that $(s,t)$ is an emulator edge and so after the failure of $F_{(u,v)}$ must have weight $\dist_{G \setminus F_{(u,v)}}(s,t)$.

But this means that the algorithm would not have added $(u,v)$ due to fault set $F_{(u,v)}$, which contradicts the definition of $(u,v)$ and $F_{(u,v)}$.
Hence if $(u, v)$ is added then $F_{(u,v)}$ cannot be mass-avoiding, which implies the lemma.
\end{proof}

We will proceed by converting this to a bound on the number of paths that we need to declare \emph{concentrated} at each scale:

\begin{lemma} \label{lem:cspathcount}
With high probability, for every spanner edge $(u, v)$ and $1 \le j \le k$, the number of concentrated splittable $j$-paths completed by $(u, v)$ is $\le \frac{d^{j-1}}{4k}$.
\end{lemma}
\begin{proof}
Let $(u,v) \in E(H^{(sp)})$.   
We will say that a node pair $(s, t) \in \Psi^j(u, v)$ is \emph{saturated} if, just before $(u, v)$ is added, we have
$$C^j_{(s, t)} \ge \tau^{\left\lceil \frac{j-1}{2} \right\rceil} - j\tau^{\left\lceil \frac{j-3}{2} \right\rceil} = \Omega\left( \tau^{\left\lceil \frac{j-1}{2} \right\rceil}\right)$$
(where the last equality is by choosing polylogs in $\tau$ large enough that $\tau \gg j$).  We know from Lemma~\ref{lem:countermassk} that with high probability, $\sum_{(s, t) \in \Psi^j(u, v)} C^j_{(s, t)} = \Oish\left(f d^{j-1} \right)$ just before $(u,v)$ is added. 
Thus with high probability, just before $(u, v)$ is added the number of saturated pairs is at most
$$\Oish\left( \frac{fd^{j-1}}{\tau^{\left\lceil\frac{j-1}{2}\right\rceil}} \right).$$
Meanwhile, by definition of dispersion, for any $1 \le h \le j$ there can be only $\tau^{\left\lceil \frac{h-2}{2} \right\rceil}$ total $s \leadsto u$ SALAD $(h-1)$-paths, and there can be only $\tau^{\left\lceil \frac{j-h-1}{2} \right\rceil}$ total $v \leadsto t$ SALAD $(j-h)$-paths. 
Thus, for any pair $(s, t) \in \Psi^j(u, v)$, the number of splittable $s \leadsto t$ $j$-paths completed by $(u, v)$ is
$$\sum \limits_{\text{even } h=2}^j \tau^{\left\lceil\frac{h-2}{2}\right\rceil} \cdot \tau^{\left\lceil \frac{ j-h-1 }{2}\right\rceil} \le j \tau^{\left\lceil \frac{j-3}{2} \right\rceil}$$
where we sum only over even $h$ because, due to the alternating property, the heaviest edge $(u, v)$ along a path must be even-numbered.

So for each \emph{unsaturated} node pair $(s, t) \in \Psi^j(u,v)$, all splittable $s \leadsto t$ paths completed by $(u, v)$ are dispersed.  This is because by the definition of saturated and the definition of the counters, there are less than $\tau^{\left\lceil \frac{j-1}{2} \right\rceil} - j\tau^{\left\lceil \frac{j-3}{2} \right\rceil}$ SALA paths $s\leadsto t$ $j$-paths before $(u,v)$ is added, and by the above calculation, adding $(u,v)$ adds an additional at most $j \tau^{\left\lceil \frac{j-3}{2} \right\rceil}$ splittable SALA $s \leadsto t$ $j$-paths.  Hence all of these new paths can be dispersed, as  we will still have at most $\tau^{\left\lceil \frac{j-1}{2} \right\rceil}$ SALAD $s \rightarrow t$ $j$-paths.  

On the other hand, for each \emph{saturated} node pair $(s, t) \in \Psi^j(u, v)$, we might need to declare all of the new splittable $s \leadsto t$ SALA $j$-paths paths to be concentrated.  We showed that there are at most $j \tau^{\left\lceil \frac{j-3}{2} \right\rceil}$ new such paths for each such $(s,t)$.  Hence the total number of splittable paths completed by $(u, v)$ that are declared concentrated is
$$\le \Oish\left( \frac{fd^{j-1}}{\tau^{\left\lceil \frac{j-1}{2} \right\rceil}} \right) \cdot j \tau^{\left\lceil \frac{j-3}{2} \right\rceil} = d^{j-1} \cdot \Oish\left( \frac{fj}{\tau} \right).$$
By setting the implicit polylogs in $\tau = \Oish(f)$ high enough, and using that $k \le \log n$ and $j = \Oish(1)$, we can ensure that the latter term is at most $\frac{1}{4k}$, and the lemma follows.
\end{proof}

Our next step towards our counting lemma roughly follows the proof of Lemma \ref{lem:3spanedge}.
We pass from $H^{(sp)}$ to a random induced subgraph $G'$ by keeping each node independently with probability $d^{-1}$, and deleting the others.
Let us say that an edge $(u, v) \in E(G')$ is \emph{clean} if:
\begin{itemize}
    \item None of the nodes in its associated fault set $F_{(u, v)}$ survive in $G'$, and
    \item For every $1 \le j \le k$ and every simple concentrated splittable $j$-path $\pi$ completed by $(u, v)$, $\pi$ does \emph{not} survive in $G'$. 
\end{itemize}

\begin{lemma} \label{lem:clean}
For any edge $(u, v) \in H^{(sp)}$, assuming that the high probability event of Lemma~\ref{lem:cspathcount} occurs, we have
$$\Pr\left[ (u, v) \text{ is clean} \ \mid \ u, v \text{ both survive in } G' \right] \ge 1/2.$$
\end{lemma}
\begin{proof}
First, since $cf \le d$, by choice of large enough constant $c$, the probability that there is some node from $F_{(u,v)}$ in $G'$ is at most 
\[
1 - \left(1-\frac1d\right)^f \leq 1/4.
\]

Fix some value of $1 \le j \le k$.
By Lemma \ref{lem:cspathcount}, there are $\le \frac{d^{j-1}}{4k}$  simple concentrated splittable $j$-paths completed by $(u, v)$. 
Conditional on $(u, v)$ itself surviving in $G'$, there are $j-1$ other nodes in each such path $\pi$, so it survives with probability $d^{-(j-1)}$.
So in expectation, $\le \frac{1}{4k}$ simple concentrated splittable $j$-paths survive.
By Markov's inequality the probability that any such paths survive is at most $\frac{1}{4k}$.
By a union bound over all choices of $j$, the probability that any concentrated splittable $1 \le j \le k$ path survives is at most $1/4$.

By a union bound on the previous two parts, $(u, v)$ is clean with probability at least $1/2$.
\end{proof}

We now pass from $G'$ to another graph $G''=(V'', E'')$ using the following two steps:
\begin{itemize}
    \item Delete all edges from $G'$ that are not clean, and then
    \item  $V''$ is defined by dividing each node $v \in V'$ into nodes $\{v_i\}$, where each node $v_i$ corresponds to one of the buckets $B_v^i$ originally associated to the node $v$, used in the definition of locality.
    Then $E''$ is defined as follows: for each edge $(u, v)$ which has survived (both endpoints are in $G'$ and it is clean), add an edge in $E''$ between the copies $(u_h, v_i)$ in $G''$, where the edge is in the $h^{th}$ bucket of $u$ and the $i^{th}$ bucket of $v$.
\end{itemize}

In the following let $\sigma$ be the number of SALAD $k$-paths in $H^{(sp)}$ and let $\sigma''$ be the number of edge-simple alternating paths in $G''$ (i.e., each edge is used at most once).
We will prove two inequalities relating $\sigma$ to the expectation of $\sigma''$, analogous to the ones used internally in Lemma \ref{lem:3spanedge}.
The first is:

\begin{lemma} \label{lem:sigmaupper}
$\mathbb{E}[\sigma''] \le \frac{\sigma}{d^{k+1}}$.
\end{lemma}
\begin{proof}
For every path $\pi \in G''$, let $s(\pi)$ denote the corresponding path in $G'$ (the path which uses the same edges).  Note that $s(\pi)$ is also a path in $H^{(sp)}$, and that if $\pi' \neq \pi$, then $s(\pi') \neq s(\pi)$.  We first argue that if $\pi$ is an edge-simple alternating path in $G''$, then $s(\pi)$ is a SALAD path in $G'$. 
\begin{itemize}
    \item (\textbf Simple) Suppose towards contradiction that $s(\pi)$ is not (node-)simple in $G'$.
    Since $\pi$ is edge-simple so is $s(\pi)$, and hence $s(\pi)$ contains a cycle subpath $C$ of length $j \le k$.
    The edge $(u, v)$ that completes $C$ must include a node in $C$ in its associated fault set $F_{(u, v)}$, since (exactly as in Lemma \ref{lem:3spanedge}) there is a $u \leadsto v$ path of length $\le k \cdot w(u, v)$ going around the cycle $C$.
    So either none of the nodes in $F_{(u, v)}$ survive in $G'$ (in which case $C$ is not in $G'$), or else $(u, v)$ is unclean and so the corresponding edge does not survive in $G''$.
    In either case, at least one edge of $\pi$ is missing from $G''$, contradicting the definition of $\pi$.
    
    \item (\textbf Alternating) Since $\pi$ is alternating in $G''$, clearly $s(\pi)$ is also alternating in $G'$ since their edges correspond.
    
    \item (\textbf Local) Due to the node-dividing step when we move from $G'$ to $G''$, any non-local path in $G'$ does not survive in $G''$.  Hence $s(\pi)$ must be local.

    \item (\textbf Avoids Faults) If $s(\pi)$ contains both an edge $(u, v)$ and a node in its associated fault set $F_{(u, v)}$, then we will either delete this node when moving from $H^{(sp)}$ to $G'$, or else $(u, v)$ is unclean and so we will delete this edge when moving from $G'$ to $G''$.  In either case, $\pi$ would not exist in $G''$, contradicting the definition of $\pi$
    
    \item (\textbf Dispersed) By the previous four bullet points, $s(\pi)$ is SALA.
    We proceed using a proof by minimal counterexample.
    Suppose towards contradiction that $s(\pi)$ is not dispersed, and let $q \subseteq s(\pi)$ be the shortest subpath of $s(\pi)$ that is SALA but not dispersed.
    Thus $q$ is splittable, since the two subpaths that emerge after the split are both shorter than $q$.
    
    Let $(u, v)$ be the heaviest edge in $q$.
    If $q$ does not survive in $G'$, then clearly $\pi$ is not a path in $G''$, which is a contradiction.  Hence $q$ is still a path in $G'$.  But then by definition the edge $(u, v)$ is unclean in $G'$.  Hence the equivalent edge would not exist in $G''$, contradicting the fact that $\pi$ is a path in $G''$.
    This completes the contradiction; thus, $s(\pi)$ must be dispersed.
\end{itemize}

So we have that every edge-simple path in $G''$ corresponds to a SALAD path in $G'$, and this correspondence is injective.  Thus $\sigma''$ is at most the number of SALAD $k$-paths in $G'$.  We can upper bound the expectation of this quantity by observing that (by simplicity) any SALAD $k$-path in $H^{(sp)}$ has $k+1$ distinct nodes, and each of these nodes survives in $G'$ with probability $d^{-1}$, so each SALAD $k$-path in $H^{(sp)}$ survives in $G'$ with probability $d^{-(k+1)}$.  Hence $\mathbb{E}[\sigma''] \leq \frac{\sigma}{d^{k+1}}$.  
\end{proof}

Our second inequality needs the following lemma, which was proved implicitly in \cite{BDR22} (it is a generalization of the ``intermediate counting lemma'' of~\cite{BDR22}; we include the proof here for completeness).
\begin{lemma} \label{lem:weakcounting}
Any $n$-node graph with $m > kn$ edges has at least $m - kn$ edge-simple alternating $k$-paths.
\end{lemma}
\begin{proof}
The weak counting lemma of~\cite{BDR22} implies that any $n$-node graph with at least $kn$ edges has at least $1$ edge-simple alternating $k$-path.  So consider the following process: pick an edge that is contained in at least one edge-simple alternating $k$-path, remove it, and repeat until there are no more edge-simple alternating $k$-paths.  Since we remove the edge we find in each iteration, there are at least as many edge-simple alternating $k$-paths as there are iterations. And there are at least $m-kn$ iterations, since as long as at least $kn$ edges remain in the graph, there is still at least $1$ edge-simple alternating $k$-path.  Hence there are at least $m-kn$ edge-simple alternating $k$-paths in total.
\end{proof}

Using this, we prove:
\begin{lemma} \label{lem:sigmalower}
$\mathbb{E}[\sigma''] = \Omega\left(\frac{n}{d}\right)$.
\end{lemma}
\begin{proof}
The number of buckets over all nodes in $H^{(sp)}$ is $O(|E(H^{(sp)})|/(kd))$, with as small an implicit constant as we like by setting the constant in the definition of $b$ to be as large as we need. 
Each bucket survives as a node in $G''$ with probability $d^{-1}$, since a bucket survives if and only if its corresponding node survives. 
So the expected number of nodes in $G''$ is at most
$$\mathbb{E}[|V(G'')|] \leq \frac{c'|E(H^{(sp)})|}{kd^2}$$
for as small a constant $c'$ as we want.  
Any particular edge in $H^{(sp)}$ survives in $G'$ with probability $d^{-2}$, and then it is clean (and thus survives in $G''$) with probability $\ge 1/2$ by Lemma~\ref{lem:clean}.
So the expected number of edges in $G''$ is at least
$$\mathbb{E}[|E(G'')|] \ge \frac{|E(H^{(sp)})|}{2d^2}.$$
By pushing $c'$ sufficiently small, it follows from Lemma \ref{lem:weakcounting} that the expected number of edge-simple alternating $k$-paths in $G''$ is
\begin{align*}
\mathbb{E}[\sigma''] &\geq \mathbb{E}[|E(G'')| - k |V(G'')|] = \mathbb{E}[|E(G'')|] - k \mathbb{E}[|V(G'')|] \\
&\geq \frac{|E(H^{(sp)})|}{2d^2} - k\frac{c'|E(H^{(sp)})|}{kd^2} \geq \Omega\left( \frac{|E(H^{(sp)})|}{d^2} \right) \\
&= \Omega\left(\frac{n\delta}{d^2}\right) = \Omega\left(\frac{kn}{d} \right) = \Omega\left(\frac{n}{d}\right),
\end{align*}
where the last line is since we assume $H^{(sp)}$ has average degree $\delta \ge cdk$.
\end{proof}

Finally, putting the pieces together: by Lemmas \ref{lem:sigmaupper} and \ref{lem:sigmalower}, we have
$$\Omega\left(\frac{n}{d}\right) = \mathbb{E}[\sigma''] \le \frac{\sigma}{d^{k+1}}$$
and so, rearranging, we get
$$\sigma = \Omega\left(nd^k\right)$$
which proves our counting lemma.

\section{Polynomial Time} \label{sec:polytime}
Unfortunately, Algorithm~\ref{alg:general-k} does not run in polynomial time.  This is because of the ``if'' condition: we need to check whether there is some $F \subseteq V \setminus \{u,v\}$ with $|F| \leq f$ such that $\dist_{H \setminus F}(u,v) > (2k-1) \cdot w(u,v)$.  The obvious way of doing this takes $\Omega(n^f)$ time in order to check all possible fault sets, which is not polynomial if $f$ is superconstant (which is the interesting case, as we are studying $f$-dependence).  One might hope to design a polynomial-time algorithm to perform this check, but even special cases of this problem are NP-hard: if the graph is unweighted then this is equivalent to the \textsc{Length-Bounded Cut} problem, which is known to be NP-hard~\cite{BEHKSS06}. 

The same problem arises in the context of VFT spanners, where the best-known algorithm is the VFT-greedy algorithm (exactly Algorithm~\ref{alg:general-k} but without adding emulator edges).
This algorithm was shown to produce spanners of existentially optimal size in~\cite{BP19,BDPW18}, but the same running time issue was present.  This issue was recently resolved by~\cite{DR20,BDR21}, who showed how to slightly change the greedy algorithm to make it polytime.  We show that we can adapt the techniques of~\cite{DR20} to the VFT emulator setting, obtaining a polynomial-time algorithm with size bounds that are only polylogarithmically worse than the exponential time algorithm (and we have suppressed our polylogs with $\Oish(\cdot)$ notation anyway).  

The main idea of~\cite{DR20} was to replace the $n^f$ time check with an \emph{approximation algorithm} for the \textsc{Length-Bounded Cut} problem on \emph{unweighted} graphs, and then show that the approximation and the restriction to unweighted graphs did not cost us much even when applied to weighted graphs.  We follow this approach, replacing the line
\begin{center}
``if there is $F \subseteq V \setminus \{u,v\}$ of size $|F| \leq f$ with $\dist_{H \setminus F}(u,v) > (2k-1) \cdot w(u,v)$''
\end{center}
in Algorithm~\ref{alg:general-k} with a new condition:
\begin{center}
``if \ffs$(G,H,u,v,k,f)$ returns YES''
\end{center}
where \ffs\ is a subroutine described below.
Intuitively, the algorithm \ffs\ is supposed to be an approximate version of our previous check.  We now give this algorithm and prove the guarantees that it provides, and then show how changing Algorithm~\ref{alg:general-k} to use \ffs\ affects the proofs from Section~\ref{sec:general-k}.

\subsection{The Distinguishing Algorithm}
Consider the following algorithm.  
\FloatBarrier
\begin{algorithm}
Consider all edges in $H$ and $G$ to be \emph{unweighted}.   So the weight of an emulator edge $(s,t)$ under fault set $f$ is equal to the number of hops between $s$ and $t$ in $G \setminus F$\;
\lIf{$(u,v) \in E(H)$}{\Return{NO}}
Initialize $F = \emptyset$\;
\While{$\dist_{H \setminus F}(u,v) \leq 2k-1$}{
Let $P$ be a shortest $u-v$ path in $H \setminus F$\;
Let $P'$ be the (possibly non-simple) path in $G \setminus F$ obtained by replacing every emulator edge in $P$ with a shortest path between its endpoints in $G \setminus F$\;
Add all vertices in $P' \setminus \{u,v\}$ to $F$\;
}
\leIf{$|F| \leq (2k-2) f$}{\Return{YES}}{\Return{NO}}

\caption{\ffs$(G,H,u,v,k,f)$}
\label{alg:ffs}
\end{algorithm}
\FloatBarrier

We claim that this is effectively a polynomial-time $(2k-2)$-approximation algorithm for the problem of finding the smallest $F$ that makes $\dist_{H \setminus F}(u,v) > 2k-1$ (recall that in this algorithm we use \emph{unweighted} graph edges).  

\begin{lemma}
\ffs\ takes polynomial time.
\end{lemma}
\begin{proof}
Each iteration of the while loop involves a shortest path computation for $u-v$, and then a shortest path computation for each emulator edge to update its weight appropriately.  Hence each iteration takes polynomial time if we use a polynomial-time shortest-path algorithm.  In each iteration we add at least one node to $F$ (since the path $P'$ must have at least two edges or else we would have immediately returned NO), and hence the number of iterations is at most $O(n)$.  Thus the total running time is polynomial.
\end{proof}

\begin{lemma} \label{lem:efficiency}
If $\dist_{H \setminus F'}(u,v) \leq 2k-1$ for all $F' \subseteq V \setminus \{u,v\}$ with $|F'| \leq (2k-2)f$, then \ffs$(G,H,u,v,k,f)$ returns NO.
\end{lemma}
\begin{proof}

We prove the contrapositive.  Suppose that \ffs$(G,H,u,v,k,f)$ returns YES.  Then by the definition of the algorithm, the set $F$ that it found has $\dist_{H \setminus F}(u,v) > 2k-1$ with $|F| \leq (2k-2)f$.
\end{proof}
\begin{lemma} \label{lem:soundness}
If there exists an $F' \subseteq V \setminus \{u,v\}$ with $|F'| \leq f$ such that $\dist_{H \setminus F'}(u,v) > 2k-1$, then \ffs$(G,H,u,v,k,f)$ returns YES.
\end{lemma}
\begin{proof}
Let $F' \subseteq V \setminus \{u,v\}$ with $|F'| \leq f$ such that $\dist_{H \setminus F'}(u,v) > 2k-1$, and let $F$ be the fault set created by the algorithm.  We need to prove that $|F| \leq (2k-2)f$, since that will imply that the algorithm returns YES.  Let $P$ be a path that the algorithm found in some iteration, and let $P'$ be the associated path in $G$ (so the algorithm added all vertices of $P'$ to $F$ other than $u,v$).  

Clearly the length of $P'$ is at most the length of $P$, since we just replaced every emulator edge by a graph path of the exact same length, and we know from the algorithm that the length of $P$ is at most $2k-1$.  Hence $P'$ has at most $2k-2$ vertices other than $u,v$.  Thus in every iteration $|F|$ grows by at most $2k-2$.

Since $P$ has length at most $2k-1$ and $\dist_{H \setminus F'}(u,v) > 2k-1$, it must be the case that $F' \cap P \neq \emptyset$.  On the other hand, if $P_1$ and $P_2$ are two \emph{different} paths found by the algorithm (so they were found in different iterations), then $P_1 \cap P_2 = \{u,v\}$ (since if $P_1$ appeared in an earlier iteration then all vertices other than $u,v$ in $P_1$ were added to $F$, so were not available for $P_2$).  Thus in every iteration we add at least one vertex from $F'$ to $F$ that was not already in $F$.  Hence the number of iterations is at most $|F'| \leq f$, and so $|F| \leq (2k-2)f$ as required.
\end{proof}

\subsection{Effect on the Analysis of Algorithm~\ref{alg:general-k}}

We now discuss what changes from Section~\ref{sec:general-k} if we replace the if condition in Algorithm~\ref{alg:general-k} with \ffs.  We give only a sketch of this changed analysis, since it simply involves repeating the analysis of Section~\ref{sec:general-k} but keeping track of an extra $O(k)$ factor.

First, it is easy to see that correctness (Lemma~\ref{lem:stretch-k}) still holds thanks to the greedy ordering and Lemma~\ref{lem:soundness}.  If $(u,v) \not \in E(H)$ then Lemma~\ref{lem:soundness} implies that for every fault set $F$ of size at most $f$, there is a $u-v$ path in $H \setminus F$ of length at most $2k-1$ in the \emph{unweighted} version.  But every edge already in $H$ when we are considering adding $(u,v)$ has weight less than $w(u,v)$, and hence this path must also have length at most $(2k-1) \cdot w(u,v)$ in the real $H$ (with weights).  

As before, for any spanner edge $(u,v)$ added by the algorithm, let $F_{(u,v)}$ denote the fault set that caused us to add it.  The main difference in our new polynomial time algorithm is that rather than knowing $|F_{(u,v)}| \leq f$, we just know from Lemma~\ref{lem:efficiency} that $|F_{(u,v)}| \leq (2k-2)f$.  So we need to trace how this change propagates through Section~\ref{sec:general-k}.  The main change is that our definition of $d$ in Lemma~\ref{lem:final_size} will be multiplied by $O(k)$.  Since $k \leq O(\log n)$, this means that the final bound on $|E(H^{(sp)})|$ in Lemma~\ref{lem:final_size} will be unchanged. 

The bound on the number of emulator edges (Lemma~\ref{lem:general_em_size}) continues to hold without change, since it simply involves counting the number of local $j$-paths.  The statement of the Lemma~\ref{lem:counting} is changed online slightly, by changing the assumption that $d \geq cf$ with the assumption that $d \geq cfk$.
The proof of the spanner edge lemma (Lemma~\ref{lem:final_size}) is also unchanged since it is just calculations assuming the counting lemma (Lemma~\ref{lem:counting}); we simply need to change $d$ to have an extra factor of $O(k)$, but this is absorbed by the polylog($n$) factors since $k \leq O(\log n)$.

Now consider Lemma~\ref{lem:countermassk}, the main ``counter mass" lemma.  We do not change the statement as written, but note that there will be an extra factor of $k$ on the right hand side, but this is hidden by the $\Oish$ notation.  Since $F_{(u,v)}$ could have size up to $(2k-2)f$, we need to union bound over all possible fault sets of that size, rather than of $f$.  Hence we need to change our definition of mass-avoiding to be 
\[
\sum_{(s,t) \in \Psi(u,v,F)} C^j_{(s,t)} > c k fd^{j-1} \log n
\]
(i.e., we add an extra factor of $(2k-2)$ to the right hand side, which we simply write as $k$ due to the inclusion of the constant $c$).  Once we make this change, we can union bound over all fault sets of size at most $(2k-2)f$ to get that every mass-avoiding $F$ (of size at most $(2k-2)f$ is good for $(u,v)$ with high probability.  The rest of the proof is unchanged: this means that if $F_{(u,v)}$ it is also good, but that would imply that the algorithm would not have added $(u,v)$ since $F_{(u,v)}$ would not have actually forced us to (contradicting the definition of $F_{(u,v)}$).  

Now consider Lemma~\ref{lem:cspathcount}, the bound on the number of concentrated splittable $j$-paths completed by each spanner edge.  This goes through without change, since the statement of Lemma~\ref{lem:countermass} is unchanged.  

We then define $G'$, clean nodes, and $G''$ as before.  The statement of Lemma~\ref{lem:clean} is unchanged, and the proof only has to change by noticing that the probability that there is some node from $F_{(u,v)}$ in $G'$ is at most $1 - \left(1-\frac{1}{d}\right)^{(2k-2)f}$, since now $|F_{(u,v)}|$ can be up to $(2k-2)f$.  But since $d$ is larger by a factor of $O(k)$, this cancels out and the remaining calculation is identical. Lemmas~\ref{lem:sigmaupper}, \ref{lem:weakcounting}, and \ref{lem:sigmalower} are unchanged, since they just use the previous lemmas (whose statements are unchanged) as black boxes.  Hence they can be combined as before, proving Lemma~\ref{lem:counting} (the main counter lemma).

\section{Lower Bound} \label{sec:LB}
In this section we show that our algorithm is optimal with respect to dependence on $f$.  Interestingly, our lower bound construction for $f$-VFT $(2k-1)$-emulators is essentially the same as the lower bound for $f$-EFT (edge fault-tolerant) $(2k-1)$-spanners from~\cite{BDPW18} (the analysis is more complex, though, as we must account for the extra power of emulator edges).

\begin{conjecture}[Erd\H{o}s girth conjecture~\cite{erdHos1964extremal}] \label{con:girth} 
For every positive integer $k$, there exists an infinite family of graphs on $n$ nodes with $\Omega(n^{1+1/k})$ edges and girth $2k+2$.
\end{conjecture}

We first prove our lower bound for the special case of $k = 2$ (stretch $3$), then handle the general case of $k \geq 3$.

\LBsmall*
\begin{proof}
Let $G = (V, E)$ be a girth conjecture graph with $k=2$, i.e., a graph with girth at least $6$ and at least $\Omega(|V|^{3/2})$ edges.  This setting of the girth conjecture has actually been proved~\cite{Wenger91}, so we use this construction.  We construct a new graph $G' = (V', E')$ as follows: let $t = \lfloor f/4 \rfloor$, let $V' = V \times [t]$, and let $E' = \{\{(u,i), (v,j)\} : \{u,v\} \in E, i,j, \in [t]\}$.  In other words, each edge $\{u,v\}$ of $G$ is replaced by a complete bipartite graph between two sets of copies (of size $t$).  For every $u \in V$, let $B_u = \{(u,i) : i \in [t]\}$ be the copies of $u$ in $V'$.  Note that 
\begin{align*}
    |E'| & = t^2 |E| \geq \Omega\left( f^2 |V|^{3/2}\right) = \Omega\left(f^2 \left(\frac{|V'|}{f}\right)^{3/2}\right) = \Omega\left(f^{1/2} |V'|^{3/2}\right),
\end{align*}
so we simply need to show that every $f$-VFT $(2k-1)$-emulator for $G'$ has at least $\Omega(|E'|)$ edges.  In the spanner setting we would accomplish this by showing that every edge in $E'$ must be in every spanner, but in the emulator setting this is not true: emulator edges can be used to replace spanner edges.  Instead, we must show that the number of emulator edges in an emulator must be roughly equal to the number of edges of $G'$ which are not in the emulator.

Let $H = (V_H, E_H)$ be an $f$-fault-tolerant $3$-emulator for $G'$, and let $e = \{(u,i), (v,j)\}$ be some edge in $E'$.  We say that an emulator edge $e' \in E_H$ \emph{protects} $e$ if there is a $2$-path between the endpoints of $e'$ in $G'$ that uses $e$ (such a path will intuitively correspond to a path of length $3$ in $H$, since the emulator edge will have length $2$ under the fault sets that we care about).  We say that $e$ is \emph{$\beta$-protected} if there are at least $\beta$ emulator edges that protect $e$.  

We first claim that if $e$ is not $f/4$-protected then $e \in E_H$.  So suppose for contradiction that $e$ is not $f/4$-protected but $e \not\in E_H$.  Let $F$ consist of all endpoints of emulator edges that protect $e$ other than $(u,i)$ and $(v,j)$, together with all vertices in $B_u \cup B_v$ other than $(u,i)$ and $(v,j)$.  Clearly $|F| \leq 4(f/4) = f$, and clearly $e \in G' \setminus F$ and hence $\dist_{G' \setminus F}((u,i), (v,j)) = 1$.  
 
Since $G$ has girth at least $6$: it is easy to see that any path from $(u,i)$ to $(v,j)$ in $H \setminus F$ that uses only edges in $G' \setminus F$ (i.e., does not use emulator edges) must have length at least $5$.  This is because any such path must essentially follow a $u-v$ path in $G$ that does not use the edge $\{u,v\}$.  

On the other hand, consider some path from $(u,i)$ to $(v,j)$ in $H \setminus F$ that does use an emulator edge $e' = \{(s, i'), (t,j')\}$.  Then by the definition of $F$, we know that $e'$ does not protect $e$.  But then it is easy see that the length of this path is strictly larger than $3$. 
Thus $\dist_{H \setminus F}((u,i), (v,j)) > 3 \cdot \dist_{G' \setminus F}((u,i), (v,j))$, contradicting our assumption that $H$ is an $f$-VFT $3$-emulator.  Hence every $e \in E' \setminus E_H$ must be $f/4$-protected.

Now consider some emulator edge $e' = \{(s,i), (t,j)\} \in E_H$.  If $\dist_G(s,t) \geq 3$ then $e'$ does not protect any edges of $G'$.  If $\dist_G(s,t) = 2$ then since $G$ has girth larger than $4$ there is a unique $2$-path $s-x-t$ between $s$ and $t$ in $G$, and so by definition $e'$ can only protect edges between $(s,i)$ and $B_x$ and between $(t,j)$ and $B_x$.  Thus every emulator edge can protect at most $2t \leq f/2$ edges of $G'$.   Since we showed that any edge in $E'$ which is not in $E_H$ must be $f/4$-protected, this implies that 
\begin{align*}
    |E_H| &= |E_H \cap E'| + |E_H \setminus E'| \geq |E_H \cap E'| + \frac{1}{f/2} \cdot \frac{f}{4} |E' \setminus E_H| \geq \Omega(|E'|),
\end{align*}
as claimed.
\end{proof}

We now modify this lower bound to hold for $k \geq 3$.  This involves using a different parameter for $t$ (basically $\sqrt{f}$ rather than $f$) and generalizing the argument, but the basic construction is the same.

\LBmain*
\begin{proof}
Let $G=(V,E)$ be a graph from the girth conjecture (Conjecture~\ref{con:girth}). We construct a new graph $G'=(V',E')$ as follows: let $t= \lceil \sqrt{f} \rceil$, let $V':= V \times [t]$, and let $E':= \{ \{(u,i), (v,j): \{ u,v \} \in E, i,j 
\in [t]\}\}$. In other words, each edge $\{u,v\}$ of $G$ is replaced by a complete bipartite graph between two sets of copies (of size $\sqrt{f}$).  For every $u \in V$, let $B_u = \{(u,i) : i \in [t]\}$ be the copies of $u$ in $V'$.  Note that 
\begin{align*}
    |E'| & = t^2 |E| \geq \Omega\left( f |V|^{1+1/k}\right) = \Omega\left(f \left(\frac{|V'|}{\sqrt{f}}\right)^{1+1/k}\right) = \Omega\left(f^{\frac12 - \frac{1}{2k}} |V'|^{1+1/k}\right),
\end{align*}
so we simply need to show that every $f$-VFT $(2k-1)$-emulator for $G'$ has at least $\Omega(|E'| / k)$ edges.  In the spanner setting we would accomplish this by showing that every edge in $E'$ must be in every spanner, but in the emulator setting this is not true: emulator edges can be used to replace spanner edges.  Instead, we must show that the number of emulator edges in an emulator must be roughly equal to the number of edges of $G'$ which are not in the emulator.

Let $H = (V_H, E_H)$ be an $f$-fault-tolerant $(2k-1)$-emulator for $G'$, and let $e = \{(u,i), (v,j)\}$ be some edge in $E'$.  We say that an emulator edge $\{(s, i'), (t, j')\} \in E_H$ \emph{protects} $e$ if there is a simple path of the form $s \leadsto u - v \leadsto t$ of length at most $k$ in $G$.  We say that $e$ is \emph{$\beta$-protected} if there are at least $\beta$ emulator edges that protect $e$.

We first claim that if $e$ is not $f/3$-protected then $e \in E_H$.  So suppose for contradiction that $e$ is not $f/3$-protected but $e \not\in E_H$.  Let $F$ consist of all endpoints of emulator edges that protect $e$ (other than $(u,i)$ and $(v,j)$ if they happen to be such an endpoint) together with all vertices in $B_u \cup B_v$ other than $(u,i)$ and $(v,j)$.  Clearly $|F| \leq 2(f/3) + 2t \leq f$, and clearly $e \in G' \setminus F$ and hence $\dist_{G' \setminus F}((u,i), (v,j)) = 1$.  

Since $G$ has girth at least $2k+2$, it is easy to see that any path from $(u,i)$ to $(v,j)$ in $H \setminus F$ that uses only edges in $G' \setminus F$ (i.e., does not use emulator edges) must have length at least $2k+1$.  This is because any such path must essentially follow a $u-v$ path in $G$ that does not use the edge $\{u,v\}$.  

On the other hand, consider some path from $(u,i)$ to $(v,j)$ in $H \setminus F$ that does use an emulator edge $e' = \{(s, i'), (t,j')\}$.  Then by the definition of $F$, we know that $e'$ does not protect $e$.   Hence the length of this path is at least 
\begin{align*}
    \dist_G(u,s) + \dist_G(s,t) + \dist_G(t,v) & \geq 2k+1.
\end{align*}

Thus $\dist_{H \setminus F}((u,i), (v,j)) > (2k-1) \cdot \dist_{G' \setminus F}((u,i), (v,j))$, contradicting our assumption that $H$ is an $f$-VFT $(2k-1)$-emulator.  Hence every $e \in E' \setminus E_H$ must be $f/3$-protected.

On the other hand, consider some emulator edge $e' = \{(s,i), (t,j)\} \in E_H$.  In order to protect \emph{any} edge, it must be the case that $\dist_G(s,t) \leq k$.  Since $G$ has girth at least $2k+2$, there is only one simple $s \leadsto t$ path in $G$ of length at most $k$.  By the definition of protection, any edge protected by $e'$ must be between some node in $B_u$ and some node in $B_v$ where $u$ and $v$ are neighbors on this path.  Hence $e'$ can protect at most $k \cdot t^2 \leq O(kf)$ edges of $G'$.   Since we showed that any edge in $E'$ which is not in $E_H$ must be $f/3$-protected, this implies that 
\begin{align*}
    |E_H| &= |E_H \cap E'| + |E_H \setminus E'| \geq |E_H \cap E'| + \frac{1}{O(kf)} \frac{f}{3} |E' \setminus E_H| \geq \Omega(|E'| / k),
\end{align*}
as claimed.
\end{proof}
\section{Additive Emulators}\label{sec:additive}
In this section we consider emulators with purely additive stretch, proving Theorems~\ref{thm:additive-2} and~\ref{thm:additive-4}. In particular, we provide simple algorithms for constructing $f$-VFT $+2$-emulators and $+4$-emulators. 

We start with the $f$-VFT $+4$-emulator; the algorithm for the $+2$-emulator will be very similar, so we will just sketch how the analysis needs to be modified.

\subsection{$+4$-emulator}
We will prove the following theorem. 

\addfour*

\paragraph{Algorithm.}
Let $d$ be a parameter which depends on $f$: if $f \leq \sqrt{n}$ then we set $d = (fn)^{1/3}$, and if $f > \sqrt{n}$ then we set $d = 2f$.  We say that a node is \emph{light} if its degree in $G$ is at most $d$, and otherwise it is \emph{dense}.

Initially our emulator $H$ is empty.
\begin{enumerate}
    \item Every light node adds all of its incident edges to $H$.
    \item For every dense node $v$, we arbitrarily choose $d$ of its neighbors in $G$ and add edges between $v$ and each of these $d$ nodes to $H$.  
    \item Set $p = \frac{12 d\ln n}{n}$.  For every pair of nodes $\{u,v\} \in V \times V$, add $\{u,v\}$ to $H$ as an emulator edge independently with probability $p$.
\end{enumerate}

\paragraph{Size Analysis.}
It is easy to see that $H$ has $O(dn \log n)$ edges: the first step adds at most $dn$ edges, the second step adds at most $dn$ edges, and a simple Chernoff bound implies that with high probability the third step adds at most $O(p \binom{n}{2}) = O(dn \log n)$ edges.  Hence the total number of edges is at most $O(dn \log n)$.  If $f \leq \sqrt{n}$, then this means that we have at most $O((fn)^{1/3} n\log n) = O(f^{1/3}n^{4/3} \log n)$ edges.  If $f > \sqrt{n}$, then we have at most $O(f n \log n)$ edges.  Hence the total number of edges is at most $O((f^{1/3} n^{4/3} + nf)\log n)$.
\paragraph{Stretch Analysis.}
We now bound the stretch. More formally,
\begin{lemma}
Let $F \subseteq V$ with $|F| \leq f$, and let $u,v \in V \setminus F$.  Then 
\[
\dist_{H\setminus F}(u,v) \leq \dist_{G \setminus F}(u,v) + 4
\]
with probability at least $1 - \frac{1}{n^{3f}}$.
\end{lemma}
\begin{proof}
Let $P = (u = x_0, x_1, \dots, x_{k-1}, x_k = v)$ be the shortest $u-v$ path in $G \setminus F$ (breaking ties arbitrarily).  If all edges of $P$ are in $E(H)$ then we are done.  Otherwise, let $i$ be the smallest integer such that $\{x_i, x_{i+1}\} \not\in E(H)$, and let $j \leq k$ be the largest integer such that $\{x_{j-1}, x_j\} \not\in E(H)$.  Then clearly $x_i$ and $x_j$ are both dense vertices, or else all of their incident edges would be in $H$.  Hence they each have more than $d$ neighbors in $G$, and so have more than $d-f \geq d/2$ neighbors from step 2 in $H \setminus F$.  Let $N(x_i)$ and $N(x_j)$ denote these nodes.

First, observe that if there is an emulator edge between some node $a \in N(x_i)$ and some node $b \in N(x_j)$ then we are done.  This is because then we would have that
\begin{align*}
    \dist_{H \setminus F}(u,v) &\leq \dist_{H\setminus F}(u, x_i) + \dist_{H \setminus F}(x_i, a) + \dist_{H \setminus F}(a,b) + \dist_{H \setminus F}(b, x_j) + \dist_{H \setminus F}(x_j, v) \\
    &\leq i + 1 + \dist_{G \setminus F}(a,b) + 1 + (k-j) \\
    &\leq i + k - j + 2 + (1 + (j-i) + 1) \\
    &= k + 4\\
    &= \dist_{G \setminus F}(u,v)+4
\end{align*}
This is also true if $a = b$, i.e., if $N(x_i) \cap N(x_j) \neq \emptyset$.  Hence we are already finished if $N(x_i) \cap N(x_j) \neq \emptyset$, so we will assume without loss of generality that $N(x_i) \cap N(x_j) = \emptyset$.  

So if we can prove that the probability that such an edge $\{a,b\}$ exists is at least $1-1/n^{3f}$ then we are finished.  The number of possible such edges is at least $|N(x_i)| \cdot |N(x_j)| \geq (d/2)^2 = d^2/4$ (since $f \leq d/2$ and $N(x_i) \cap N(x_j) = \emptyset$).  Each possible edge is added to $H$ with probability $p$ in the third step of our construction.  Hence the probability that none of these edges are added is at most:
\begin{align*}
    (1-p)^{d^2/4} &= \left(1-\frac{12d\ln n}{n}\right)^{d^2/4} \leq \exp\left(-\frac{3d^3\ln n}{n}\right)
\end{align*}

If $f \leq \sqrt{n}$ then $d = (fn)^{1/3}$, and hence the probability that we fail to get an appropriate emulator edge is at most: 
\[
\exp(-{3f \ln n}) \leq n^{-3f}
\]
  Similarly, if $f \geq \sqrt{n}$ then $d = 2f$, so the probability that we fail to get an appropriate emulator edge is at most
\[
\exp\left(-\frac{24 f^{3} \ln n}{n} \right) \leq \exp(-3 f \ln n) \leq n^{-3f}
\]

Therefore, with probability at least $1-n^{-3f}$ we have added an emulator edge that satisfies the stretch guarantee for $u,v$. 
\end{proof}

We now just need two union bounds over all possible fault sets and pairs $u,v$ to finish the stretch analysis. In particular, the union bound is over all $\binom{n}{f} n^2 \leq n^{f+2}$ choices of $F$ and $u,v$ to get that the desired stretch bound holds with high probability for all possible $F, u, v$.

\subsection{$+2$-emulator}
We now show that the same algorithm with a slightly different parameter setting and analysis leads to a $+2$-emulator. More formally, we show the following theorem.
\addtwo*

\paragraph{Algorithm.}
Let $d = (fn)^{1/2}$. Similar to the previous section, we say that a node is \emph{light} if its degree in $G$ is at most $d$, and otherwise it is \emph{dense}.

We start with an empty emulator $H$.
\begin{enumerate}
    \item Every light node adds all of its incident edges to $H$.
    \item For every dense node $v$, we arbitrarily choose $d$ of its neighbors in $G$ and add edges between $v$ and each of these nodes to $H$.  E
    \item Set $p = \frac{6d\ln n}{n}$.  For every $\{u,v\} \in V \times V$, add $\{u,v\}$ to $H$ independently with probability $p$.
\end{enumerate}

\paragraph{Size Analysis.}
The first step adds at most $dn$ edges, and the second step also adds at most $dn$ edges.  And a simple Chernoff bound implies that with high probability the third step adds at most $O(p \binom{n}{2}) \leq O(nd \log n)$ edges.

\paragraph{Stretch Analysis.}
For any pair of nodes $u,v$, consider the shortest path $P= (u = u_0,...,u_i,...,u_\ell=v)$ between $u$ and $v$ in $G \setminus F$. If all edges on the path are in $H$, we are done. Let $x=u_{i}$ be the first node on the path such that $(u_i,u_{i+1}) \not \in E(H)$, and let $y=u_j$ be the furthest node on $P$ such that $(u_{j-1},u_j) \not \in E(H)$. Then we know that $x$ is dense (as otherwise we would have added all of its incident edges), and hence has $d$ edges incident to it from step 2 of the algorithm.  At least $d - f \geq d/2$ of these neighbors remain in $H \setminus F$.  Let $N(x)$ denote these neighbors.   

It is easy to see that with high probability there is an emulator edge added between $N(x)$ and $y$ in the third step of the algorithm: the probability of not adding such an edge is at most
\begin{align*}
    (1-p)^{d/2} &= \left(1-\frac{6d\ln n}{n}\right)^{d/2} \leq \exp\left(-\frac{3d^2\ln n}{n}\right)
\end{align*}

Hence the probability that no edge is added between $N(x)$ and $y$ can be bounded by $\exp(-3f\ln n)\leq n^{3f}$. Hence with probability at least $1-1/n^{3f}$ we have added an emulator edge $(z,y)$ for some $z \in N(x)$.

If we have such an edge, then by using it and the triangle inequality we have that 
\begin{align*}
    \dist_{H \setminus F}(u,v) &\leq \dist_{H \setminus F}(u,x)+\dist_{H\setminus F}(x,z)+ \dist_{H \setminus F} (z,y)+\dist_{H \setminus F} (y,v)\\
    &\leq i+1 +\dist_{G\setminus F}(z,y)+(\ell-j)\\
    &=\ell+2= \dist_{G \setminus F}(u,v) +2.
\end{align*}

This implies that with probability at least $1-n^{3f}$ the stretch guarantee holds for $u,v$ and $F$. Now as before a union bound over all fault sets of size at most $f$ and all pairs $u,v$ ($n^{f+2}$ choices) implies that the stretch guarantee holds with high probability for all such $u,v,F$.

\bibliography{refs}

\appendix

\section{Proofs Omitted from Section \ref{sec:k=3}} \label{app:proofs}

\countermass*
\begin{proof}

Let $(u,v)$ be an edge in the input graph, and let $F \subseteq V \setminus \{u,v\}$ with $|F| \leq f$.
Consider the moment in the algorithm where we inspect $(u, v)$ and decide whether or not to add it to the emulator (note: $(u, v), F$ are arbitrary; we may or may not actually add $(u, v)$, and if we do we do not necessarily have $F = F_{(u, v)}$).
We use the following extensions of our previous definitions:
\begin{itemize}
    \item For a path $\pi$ in $H^{(sp)}$ that \emph{would} be completed if we added $(u, v)$ to the emulator, we say that $\pi$ is \emph{$F$-avoiding} if $\pi \cap F = \emptyset$. 
    
    \item $\Psi(u,v,F)$ is the set of node pairs $(s, t) \in V \times V$ such that, if we added $(u, v)$ to the emulator, it would complete at least one new middle-heavy $F$-avoiding $3$-path from $s$ to $t$.

\item We say that $F$ is \emph{mass-avoiding} for $(u,v)$ if
\[
\sum_{(s,t) \in \Psi(u,v,F)} C_{(s,t)} > c fd^2 \log n.
\]
where $c$ is some large enough absolute constant.

\end{itemize}

Note that the lemma statement is equivalent to the claim that, \emph{if} $(u,v)$ is added to $H^{(sp)}$, \emph{then} $F_{(u,v)}$ is not mass-avoiding.
We have set up these definitions for general $(u, v), F$ because our proof strategy is to take a union bound over \emph{all} possible choices of $(u, v), F$, which will thus include $F_{(u, v)}$.

We say that a mass-avoiding $F$ is \emph{good} for $(u,v)$ if (immediately prior to $(u,v)$ being considered by the algorithm) there is some  $(s,t) \in \Psi(u,v,F)$ such that $(s,t)$ is already an emulator edge in $H$.  Otherwise, we say that $F$ is \emph{bad} for $(u,v)$.

We now prove that with high probability, every mass-avoiding $F$ is good for $(u,v)$.  To see this, consider some mass-avoiding $F$.  Every middle-heavy fault-avoiding $3$-path $\pi = (s,x,y,t)$ which contributes to $C_{(s,t)}$ was completed by its middle edge $(x,y)$ (since the path is middle-heavy and the algorithm considers edges in increasing weight order) and does not intersect $F_{(x,y)}$ (since it is fault-avoiding).
So, by definition of the algorithm, when $\pi$ was completed we sampled $(s,t)$ as an emulator edge with probability $d^{-2}$.
No two such paths share the same middle edge, and hence we independently add $(s,t)$ as an emulator edge with probability $d^{-2}$ at least $C_{(s,t)}$ times.  These choices are also clearly independent for different pairs $(s,t)$ and $(s', t')$, and hence the probability that $F$ is bad is at most
\begin{align*}
\prod_{(s,t) \in \Psi(u,v,F)} \left(1-\frac{1}{d^2}\right)^{C_{(s,t)}} &
\leq \exp\left( - \frac{\sum_{(s,t) \in \Psi(u,v,F)} C_{(s,t)}}{d^2}\right) 
\leq \exp\left( -cf \log n\right) 
\leq 1/n^{f+10},
\end{align*}
where we used that $F$ is mass-avoiding and we set $c$ sufficiently large.  There are at most $n^f$ possible mass-avoiding sets $F$ (since $|F| \leq f$), so a union bound over all all of them implies that every mass-avoiding set $F$ is good for $(u,v)$ with probability at least $1-1/n^{10}$.  
We can now do another union bound over all $(u,v)$ to get that this holds for \emph{every} $(u,v) \in E$ (whether added to $H^{(sp)}$ or not) with probability at least $1-1/n^8$.  

Now consider some $(u,v) \in H^{(sp)}$.  By the above, if $F_{(u,v)}$ is mass-avoiding, then it must be good.  Hence there is some emulator edge $(s,t)$ with $(s,t) \in \Psi(u,v,F_{(u,v)})$, which implies that $(s,u), (v,t) \in H^{(sp)}$ and $s,t\not\in F_{(u,v)}$.  Thus immediately prior to adding $(u,v)$ to $H^{(sp)}$, it was the case that 
\begin{align*}
\dist_{H \setminus F_{(u,v)}}(u,v) &\leq \dist_{H \setminus F_{(u,v)}}(u,s) + \dist_{H \setminus F_{(u,v)}}(s,t) + \dist_{H \setminus F_{(u,v)}}(t,v) \\
&\leq \dist_{G \setminus F_{(u,v)}}(u,s) + \dist_{G \setminus F_{(u,v)}}(s,t) + \dist_{G \setminus F_{(u,v)}}(t,v) \\
&\leq \dist_{G \setminus F_{(u,v)}}(u,s) + \left(\dist_{G \setminus F_{(u,v)}}(s,u) + \dist_{G \setminus F_{(u,v)}}(u,v) + \dist_{G \setminus F_{(u,v)}}(t,v) \right) + \dist_{G \setminus F_{(u,v)}}(t,v) \\
&\leq 5 \cdot \dist_{G \setminus F_{(u,v)}} (u,v)\\
&= 5 \cdot w(u,v).
\end{align*}
In the above inequalities we used the triangle inequality, the fact that edges are added in increasing weight order and $(u,s)$ and $(t,v)$ have already been added and so are lighter than $(u,v)$, and the fact that $(s,t)$ is an emulator edge and so after the failure of $F_{(u,v)}$ must have weight $\dist_{G \setminus F_{(u,v)}}(s,t)$.

But this means that the algorithm would not have added $(u,v)$ due to fault set $F_{(u,v)}$, which contradicts the definition of $(u,v)$ and $F_{(u,v)}$.
Hence if $(u, v)$ is added then $F_{(u,v)}$ cannot be mass-avoiding, which implies the lemma.
\end{proof}

\kthreealgebra*
\begin{proof}
Let $m$ be the number of middle-heavy fault-avoiding $3$-paths in the final graph $H^{(sp)}$.
By Lemma \ref{lem:mfpathcount}, we have that with high probability
$$m = d^2 \left|E\left(H^{(sp)}\right)\right| + \Oish\left(fn^2 \right).$$
We condition on this event occurring in the remainder of the argument.
There are two cases, depending on which of these terms is larger.

\paragraph{Case 1: the term $d^2 \left|E(H^{(sp)}\right|$ is larger.}
In this case, by Lemma \ref{lem:3emedge}, the total number of emulator edges is
$O\left( \left| E\left(H^{(sp)}\right) \right| \right)$,
and so it suffices to bound the spanner edges.
By Lemma \ref{lem:3spanedge}, the total number of spanner edges is
$$\left| E\left(H^{(sp)}\right) \right| = O\left( \left(d^2 \left|E(H^{(sp)}\right|\right)^{1/3} n^{2/3} + nf\right).$$
If the latter term dominates then $\left| E\left(H^{(sp)}\right) \right| = O\left( nf\right)$, so we are done.
If the former term dominates, then by rearranging we get
$\left| E\left(H^{(sp)}\right) \right| = O\left( nd \right)$,
and so in this case $|E(H)| = O(nd)$.

\paragraph{Case 2: the term $\Oish\left(fn^2 \right)$ is larger.}
In this case, by Lemma \ref{lem:3emedge}, the total number of edges in $H^{(em)}$ is $\Oish\left(fn^2 d^{-2}\right)$.
By Lemma \ref{lem:3spanedge}, the total number of spanner edges in $H^{(sp)}$ is
\begin{align*}
\left| E\left(H^{(sp)}\right) \right| &= O\left( \left( \Oish\left(fn^2 \right) \right)^{1/3} n^{2/3} + nf\right) = \Oish\left( f^{1/3} n^{4/3}\right) + O\left(nf\right).
\end{align*}
So in this case, the number of edges in the final emulator $H$ is
$$\Oish\left(fn^2 d^{-2} + f^{1/3} n^{4/3}\right) + O\left(nf\right).$$

\paragraph{Putting It Together.}

In either case, the total number of edges in the final emulator $H$ is
$$\Oish\left(fn^2 d^{-2} + f^{1/3} n^{4/3}\right) + O\left(nd + nf\right).$$
Setting $d = f^{1/3} n^{1/3}$, we thus get $\Oish\left(f^{1/3} n^{4/3}\right) + O(nf)$ edges in total, as claimed.
\end{proof}

\end{document}